\documentclass[10pt]{article}
\usepackage[english]{babel}
\usepackage{color}
\usepackage{graphicx}
\usepackage{framed}
\usepackage[normalem]{ulem}
\usepackage{mathtools}
\usepackage{bbold}
\usepackage{amsmath}
\usepackage{amsthm}
\usepackage{amssymb}
\usepackage{amsfonts}
\usepackage{bbm}
\usepackage{enumerate}
\usepackage{algorithm}
\usepackage[noend]{algpseudocode}
\usepackage[utf8]{inputenc}
\usepackage[top=1 in,bottom=1in, left=1 in, right=1 in]{geometry}

\usepackage{url}

\newtheorem{theorem}{Theorem}[section]

\newtheorem{lemma}[theorem]{Lemma}

\newtheorem{corollary}[theorem]{Corollary}

\makeatletter
\renewenvironment{proof}[1][\proofname]{\par
  \vspace{-\topsep}
  \pushQED{\qed}%
  \normalfont
  \topsep3pt \partopsep3pt 
  \trivlist
  \item[\hskip\labelsep
        \itshape
    #1\@addpunct{.}]\ignorespaces
}{%
  \popQED\endtrivlist\@endpefalse
  \addvspace{6pt plus 6pt} 
}
\makeatother

\newcommand{\setmargins}[0]{\vspace*{3pt}\setlength{\itemsep}{16pt}\setlength{\parsep}{0pt}\setlength{\parskip}{-10pt}}

\newcommand{\comm}[2]{\ensuremath{\mathrm{comm}(#1,#2)}}
\newcommand{\zet}[2]{\ensuremath{\Omega_{#1,#2}}}

\title{On sets of commuting and anticommuting Paulis}
\author{Rahul Sarkar\footnote{Institute for Computational and Mathematical Engineering, Stanford University, Stanford, CA, USA}\ \footnote{This work was done during an internship at the IBM T.J. Watson Research Center, Yorktown Heights, NY, USA} and Ewout van den Berg\footnote{IBM T.J. Watson Research Center, Yorktown Heights, NY, USA}}
\date{\today}

\begin{document}

\maketitle

\abstract{
In this work we study the structure and cardinality of maximal sets of commuting and anticommuting Paulis in the setting of the abelian Pauli group. We provide necessary and sufficient conditions for anticommuting sets to be maximal, and present an efficient algorithm for generating anticommuting sets of maximum size. As a theoretical tool, we introduce commutativity maps, and study properties of maps associated with elements in the cosets with respect to anticommuting minimal generating sets. We also derive expressions for the number of distinct sets of commuting and anticommuting abelian Paulis of a given size.
}

\section{Introduction}

In this work we study properties of sets of Pauli operators such that
the elements either all pairwise commute or all pairwise
anticommute. Sets of mutually commuting Paulis arise in the theory of
quantum error correction, for instance in stabilizer
theory~\cite{gottesman1997phd}. Anticommuting Paulis arise in the
mapping of Majorana operators to qubits in fermionic quantum
computation~\cite{bravyi2002fermionic}, as well as in the design of
space-time codes for wireless
communication~\cite{calderbank2001orthogonal}.

An $n$\textit{-Pauli} operator $P$ is formed as the Kronecker product
$\bigotimes_{i=1}^{n}T_i$ of $n$ terms $T_i$, where each term $T_i$ is
either the two-by-two identity matrix $\sigma_i$, or one of the three
Pauli matrices $\sigma_x$, $\sigma_y$, and $\sigma_z$. Pauli operators
have the property that any two operators, $P$ and $Q$, either commute
($PQ=QP$) or anticommute ($PQ=-QP$). Pauli operators can be
represented as strings $\{i,x,y,z\}^n$ and commutativity between two
operators is conveniently determined by counting the number of
positions in which the corresponding string elements differ and
neither element is $i$. If the total count is even, the operators
commute, otherwise they anticommute.  Given two Pauli operators $P$
and $Q$, and arbitrary non-zero coefficients
$\alpha,\beta\in\mathbb{C}$, it is easily verified that the
commutativity of $\alpha P$ and $\beta Q$ is the same as that of $P$
and $Q$. In Section~\ref{Sec:Group} we define the notion of the
abelian Pauli group, thereby allowing us to ignore such coefficients,
which may arise when multiplying Pauli operators. We then study sets
of mutually commuting Paulis in Section~\ref{Sec:Commuting}, and sets
of mutually anticommuting Paulis in Section~\ref{Sec:Anticommuting}.
We study the number of distinct maximally commuting and anticommuting
sets in Section~\ref{Sec:Counting}. We conclude the paper with a
discussion in Section~\ref{Sec:Conclusions}.

\section{Group structure}\label{Sec:Group}

In this subsection we define the \textit{abelian Pauli} group, which
forms the foundation for the remainder of the paper. We can define the
elements of the $n$-\textit{Pauli group} $\mathcal{P}_n$ as all
possible products of $n$-Pauli operators. It is easily checked that
$\mathcal{P}_n$ is a non-abelian group of order $4^{n+1}$. The set
$K = \{I, -I, iI, -iI\}$ is a normal abelian subgroup of
$\mathcal{P}_n$, and we define the \textit{abelian $n$-Pauli group} as
the quotient group $\mathcal{P}_n / K$.  The associated canonical
quotient map will be denoted by $\pi$, i.e.,
$\pi : \mathcal{P}_n \rightarrow \mathcal{P}_n / K$, which is
surjective and is given by $\pi(g) = gK$. We will sometimes denote
$\pi(g)$ by the equivalence class $[g]$, under the quotient
map. $\mathcal{P}_n / K$ is an abelian group of order $4^{n}$, and the
order of each element of the group, other than $I$, is $2$.
Given $\mathcal{H}\subseteq\mathcal{P}_n/K$ and $P\in\mathcal{P}_n/K$
we define multiplication of the element $P$ with the set $\mathcal{H}$
as $P \ast \mathcal{H} = \{PQ : Q \in \mathcal{H} \}$. For two sets
$\mathcal{H}_1,\mathcal{H}_2 \subseteq\mathcal{P}_n/K$ we define
$\mathcal{H}_1\ast\mathcal{H}_2 = \{P_1P_2 : P_1\in\mathcal{H}_1,
P_2\in\mathcal{H}_2\}$.
Multiplications with empty sets give empty sets.  We frequently need
to take the product of all elements in a set $\mathcal{H}$, and write
$\prod\mathcal{H} := \prod_{Q \in \mathcal{H}} Q$. We define the
product of the empty set to be the identity element $I$. We denote
subsets by $\subseteq$ and proper subsets by $\subset$.

\subsection{Generators}
Let $\mathcal{H}$ be a subset of $\mathcal{P}_n/K$. A set
$\mathcal{G}\subseteq\mathcal{H}$ is a \textit{generating set} of
$\mathcal{H}$ whenever any element in $\mathcal{H}$ can be expressed
as a product of the elements in $\mathcal{G}$.  For any
$\mathcal{G} \subseteq \mathcal{P}_n / K$, we denote by
$\langle\mathcal{G}\rangle$ the \textit{generated set} of
$\mathcal{G}$; that is, all elements that can be generated by products
of the elements in $\mathcal{G}$, and thus any subset is a generating
set of a possibly larger subset of $\mathcal{P}_n / K$. The set
$\mathcal{G}$ is called a \textit{minimal generating set}, if no
proper subset of $\mathcal{G}$ generates
$\langle\mathcal{G}\rangle$. We say that the elements in minimal
generating sets are \textit{independent}.

We begin by proving some elementary properties of generating sets of
$\mathcal{P}_n / K$, in particular minimal generating sets, which are
used subsequently in this paper. The first lemma shows that generated
sets always form a subgroup of $\mathcal{P}_n / K$, and also
characterizes the sizes of minimal generating sets in relation to the
sizes of the sets generated by them.

\begin{lemma}\label{lem-generator-abelian-pauli1}
The abelian Pauli group $\mathcal{P}_n / K$ satisfies the following properties.

\begin{enumerate}[(a)]
\item If $\mathcal{G} \subseteq \mathcal{P}_n / K$ is non-empty, then
  $\langle \mathcal{G} \rangle$ is a subgroup of $\mathcal{P}_n / K$.

\item If $\mathcal{S}$ is a subgroup of $\mathcal{P}_n / K$, then
  $|\mathcal{S}| = 2^\ell$, for some $0 \leq \ell \leq
  2n$.
  $\mathcal{G}$ is a generating set of $\mathcal{S}$ iff
  $\langle \mathcal{G} \rangle = \mathcal{S}$. For minimal generating
  sets $\mathcal{G}$ it holds that $I \in \mathcal{G}$ iff
  $\mathcal{S} = \{I\}$. If $\mathcal{S}\neq \{I\}$ then a minimal
  generating set $\mathcal{G}$ of $\mathcal{S}$ always exists, and
  satisfies $|\mathcal{G}| = \ell$, and $\prod \mathcal{H} \neq I$ for
  all non-empty $\mathcal{H} \subseteq \mathcal{G}$.

\item If $\mathcal{G} \subseteq \mathcal{P}_n / K$ is a minimal
  generating set, then $|\mathcal{G}| \leq 2n$. Moreover if
  $|\mathcal{G}| \geq 2$, and $\mathcal{G}' \subset \mathcal{G}$, then
  $P \in (\mathcal{G} \setminus \mathcal{G}')$ implies
  $P \notin \langle \mathcal{G}' \rangle$.
\end{enumerate}
\end{lemma}

\begin{proof}

\textit{(a)} Take any $P \in \mathcal{G}$ and notice that $P^2 = I$,
  and so $I \in \langle \mathcal{G} \rangle$. Now if
  $P,Q \in \langle \mathcal{G} \rangle$, then both can be expressed as
  products of elements in $\mathcal{G}$, and hence $PQ$ can also be
  expressed as products of elements in $\mathcal{G}$. This shows that
  $\langle \mathcal{G} \rangle$ is a subgroup of $\mathcal{P}_n / K$,
  as the other group axioms hold automatically.

\textit{(b)} The order of $\mathcal{P}_n / K$ is $4^n$. Hence
Lagrange's theorem~\cite{ROT1999a} implies that if $\mathcal{S}$ is a
subgroup of $\mathcal{P}_n / K$ then the order of $\mathcal{S}$ must
divide $4^n$, and so $|\mathcal{S}| = 2^\ell$, for some
$0 \leq \ell \leq 2n$. If $\langle\mathcal{G}\rangle = S$, then
$\mathcal{G}$ is clearly a generating set for $S$. For the converse,
let $\mathcal{G}$ be a generating set for $S$. By definition of
generating sets we have $\mathcal{G} \subseteq \mathcal{S}$, and as
$\mathcal{S}$ is a subgroup it follows that
$\langle \mathcal{G} \rangle \subseteq \mathcal{S}$. Since
$\mathcal{G}$ generates $S$ we also have
$\mathcal{S} \subseteq \langle \mathcal{G} \rangle$, and hence
$\langle \mathcal{G} \rangle = \mathcal{S}$.
If $\mathcal{S} = \{I\}$ then the minimal generating set is
$\mathcal{G}=\{I\}$. Given a minimal generating set
$\mathcal{G}\neq\{I\}$ it holds that
$\langle \mathcal{G} \rangle = \langle \mathcal{G} \setminus \{I\}
\rangle$,
since any other element can be used to generate $I$. We now show that
for a generating set $\mathcal{G}$ of $\mathcal{S} \neq \{I\}$, we
must have that $\prod \mathcal{H} \neq I$ for all non-empty subsets
$\mathcal{H} \subseteq \mathcal{G}$.  The condition holds when
$\mathcal{H}$ has size one because $I\notin \mathcal{H}$. Otherwise,
suppose $\prod\mathcal{H}=I$, then for any $P \in \mathcal{H}$ it
holds that $P = \prod (\mathcal{H} \setminus \{P\})$ and therefore
$\langle \mathcal{G} \rangle = \langle \mathcal{G} \setminus \{P\}
\rangle$, which contradicts minimality of $\mathcal{G}$.

The minimal generating set of $\mathcal{S}\neq\{I\}$ always exists; we
start with generator set $\mathcal{G} = \{Q\}$ for any
$Q\in\mathcal{S}\setminus\{I\}$, and repeatedly update it with an
element $P\in(\mathcal{S}\setminus\langle\mathcal{G}\rangle)$ until
$\mathcal{G}$ generates $\mathcal{S}$. So it only remains to show that
$| \mathcal{G} | = \ell$, which is now equivalent to showing that
$|\langle \mathcal{G} \rangle| = 2^{|\mathcal{G}|}$ because
$\langle \mathcal{G} \rangle = \mathcal{S}$. But this is immediate
from the fact that
$\langle \mathcal{G} \rangle = \{ \prod \mathcal{H} : \mathcal{H}
\subseteq \mathcal{G}\}$
if we can show that non-empty distinct subsets
$\mathcal{H}_1, \mathcal{H}_2 \subseteq \mathcal{G}$ give rise to
distinct elements $\prod \mathcal{H}_1$, and $\prod
\mathcal{H}_2$.
Suppose for the sake of contradiction that
$\prod \mathcal{H}_1 = \prod \mathcal{H}_2$. This would imply that
$I = (\prod \mathcal{H}_1) (\prod \mathcal{H}_2) = \prod
((\mathcal{H}_1 \cup \mathcal{H}_2) \setminus (\mathcal{H}_1 \cap
\mathcal{H}_2))$,
which contradicts the fact that no subset of $\mathcal{G}$ multiplies
to $I$.

\textit{(c)} Suppose that $\mathcal{G}$ is a minimal generating set,
and $|\mathcal{G}| \geq 2n+1$. Then by (b),
$|\langle \mathcal{G} \rangle| \geq 2^{2n+1} > 4^n$. This is a
contradiction as $\mathcal{P}_n / K$ is a generating set for itself,
and its order is $4^n$.
If $|\mathcal{G}| \geq 2$, then
$\langle \mathcal{G} \setminus \{I\} \rangle = \langle \mathcal{G}
\rangle$,
and as $\mathcal{G}$ is a minimal generating set we must have
$I \notin \mathcal{G}$. Now for the sake of contradiction assume that
for $\mathcal{G}' \subset \mathcal{G}$ there exist a
$P \in \mathcal{G} \setminus \mathcal{G}'$ such that
$P \in \langle \mathcal{G}' \rangle$. Then $P$ can be expressed as a
product of elements in $\mathcal{G'}$, and thus
$\langle \mathcal{G} \rangle = \langle \mathcal{G} \setminus \{P\}
\rangle$,
which contradicts that $\mathcal{G}$ is a minimal generating
set.\vspace*{-10pt}

\end{proof}

The next lemma characterizes what happens when we take the product of
all the elements of a generated set. This lemma is used often in this
paper; for instance, it will immediately imply that a maximally
commuting subgroup, defined later, multiplies to $I$.

\begin{lemma}
\label{lem-generator-abelian-pauli2}
Let $\mathcal{G} \subseteq \mathcal{P}_n / K$ be a generating set. Then 
\begin{equation}
\prod \langle \mathcal{G} \rangle =
\begin{cases}
Q & \;\; \text{if } \mathcal{G} = \{Q\}, \\
I & \;\; \text{otherwise}.
\end{cases}
\label{eq:prod-generated-set}
\end{equation}

Thus if $\mathcal{S}$ is a subgroup of $\mathcal{P}_n / K$, and $|\mathcal{S}| \neq 2$, then $\prod \mathcal{S} = I$.
\end{lemma}

\begin{proof}
If $\mathcal{G} = \{Q\}$, and $Q \neq I$, then
$\langle \mathcal{G} \rangle = \{I, Q\}$ as $Q^2 = I$, and the
statement is true. If $\mathcal{G} = \{I\}$, then
$\langle \mathcal{G} \rangle = \{I\}$ and the statement is also
true.
For the general case, let $\mathcal{G}'= \{P_i\}_{i=1}^{\ell}$ be a
minimal generating set of $\langle \mathcal{G} \rangle$ with
cardinality $\ell \geq 2$. By Lemma
\ref{lem-generator-abelian-pauli1} (b), $I \notin \mathcal{G}'$, and
therefore $P_i\neq I$ for all $i\in[\ell]$. Denote
$\mathcal{G}_k' = \bigcup_{i=1}^k\{P_i\}$, then
$\langle\mathcal{G}_2'\rangle = \{I, P_1, P_2, P_1P_2\}$, where
$P_1 P_2$ is distinct from $I,P_1,P_2$, and the product is therefore
$I$. By induction, suppose the result holds for some
$2\leq k < \ell$, then by the fact that
$\langle \mathcal{G}_k\rangle$ is an abelian subgroup, and
$\vert\mathcal{G}_k\vert$ is even, we have
\[
\prod\langle\mathcal{G}_{k+1}\rangle = \prod\langle\mathcal{G}_{k}\rangle\cdot\prod P_{k+1}\ast\langle \mathcal{G}_{k}\rangle =
P_{k+1}^{\vert\mathcal{G}_k\vert}\prod\langle\mathcal{G}_{k}\rangle = P_{k+1}^{\vert\mathcal{G}_k\vert} = I.
\]

If $|\mathcal{S}| = 1$, then $\mathcal{S} = \{I\}$ implying
$\prod \mathcal{S} = I$. By Lemma
\ref{lem-generator-abelian-pauli1}(b), $|\mathcal{S}|$ is a power of
$2$, so assume that $|\mathcal{S}| \geq 4$. As $\mathcal{S}$ is a
generating set for itself, the result above implies that
$\prod \mathcal{S} = I$.
\end{proof}

\subsection{Commutativity} 
Each element of $\mathcal{P}_n / K$ is an equivalence class containing
exactly four Pauli operators. In a slight deviation to standard
terminology, we say that two elements $P,Q \in \mathcal{P}_n/K$
\emph{commute} (\emph{anticommute}) whenever any chosen representative
of $P$ commutes (anticommutes) with any chosen representative of
$Q$. It is easily verified that this is a well defined notion, that
does not depend on the choice of the representatives.  Throughout the
remainder of the paper we exclusively use the terms `commute' or
`anticommute' to refer to this notion, rather than that of the group
operation on $\mathcal{P}_n / K$. With this convention, we say that a
subset $\mathcal{H}\subseteq\mathcal{P}_n/K$ is \textit{commuting
  (anticommuting)}, if no two distinct elements $P,Q\in\mathcal{H}$
anticommute (commute). Given any $P\in\mathcal{P}_n/K$, we say that
$P$ commutes with $\mathcal{H}$ if it commutes with all elements in
$\mathcal{H}$, and likewise for anticommutes.

Given $P,Q\in\mathcal{P}_n/K$ we define the commutativity function $\comm{P}{Q}$ such that
\[
\comm{P}{Q} = \begin{cases}
\phantom{-}1 & \mbox{if $P$ and $Q$ commute,}\\
-1 & \mbox{otherwise.}\end{cases}
\]
For any set $\mathcal{H} \subseteq \mathcal{P}_n / K$ and element
$P \in \mathcal{P}_n / K$, we define the \textit{commutativity map} of
$P$ with respect to $\mathcal{H}$ as
$\zet{P}{\mathcal{H}} : \mathcal{H} \rightarrow \{1, -1\}$, such that
$\zet{P}{\mathcal{H}}(Q) = \comm{P}{Q}$ for all $Q\in\mathcal{H}$. It
is clear that if $|\mathcal{H}| = k$, then there are a maximum of
$2^{k}$ distinct commutativity maps. We will say that the
commutativity map is \textit{all commuting (all anticommuting)}, if
$P$ commutes (anticommutes) with all elements in $\mathcal{H}$.

We now provide an important lemma that states that, given a minimal
generating set $\mathcal{G}$, each commutativity map with respect to
$\mathcal{G}$ is equally likely. This key fact is used in several
results proved later.

\begin{lemma}\label{lem-pattern-generators}
Let $\mathcal{G}\subseteq \mathcal{P}_n/K$ be a minimal generating
set with $\mathcal{G}\neq\{I\}$ and $\vert\mathcal{G}\vert=k$. Then
each of the $2^k$ possible commutativity maps with respect to
$\mathcal{G}$ is generated by $4^n/2^k$ distinct elements
$P\in\mathcal{P}_n/K$.
\end{lemma}
\begin{proof}
Let $\mathcal{G}'$ be a minimal generating set for $\mathcal{P}_n/K$
such that $\mathcal{G}\subseteq\mathcal{G}'$, which exists by the
proof of Lemma~\ref{lem-generator-abelian-pauli1}(b). Given any
distinct elements $P,Q\in\mathcal{P}_n/K$, then the commutativity
maps with respect to $\mathcal{G}'$ must differ, otherwise $PQ$
commutes with all elements in $\mathcal{G}'$, and therefore with
$\langle \mathcal{G}'\rangle = \mathcal{P}_n/K$. But that would mean
that $PQ=I$, and therefore that $P=Q$. It follows that
$\zet{P}{\mathcal{G}'} \neq \zet{Q}{\mathcal{G}'}$ whenever
$P\neq Q$, and shows that there are $4^{n}$ distinct maps
$\zet{\bullet}{\mathcal{G}'}$. We can iteratively remove elements
from $\mathcal{G}'$ to arrive at $\mathcal{G}$. Each time we remove
an element we collapse maps that differ only in the commutativity
with the removed element, which means that the number of different
maps is halved, but their occurrence is doubled. The result then
follows directly.
\end{proof}

\subsection{Decomposition}

 We can decompose any set $\mathcal{S}\subseteq\mathcal{P}_n/K$, with $n \geq 2$, as
\begin{equation}\label{Eq:DecomposeXYZ}
\mathcal{S} =  (\sigma_i \otimes \mathcal{C}_i) \cup (\sigma_x \otimes
\mathcal{C}_x) \cup (\sigma_y \otimes \mathcal{C}_y) \cup (\sigma_z
\otimes \mathcal{C}_z),
\end{equation}
with possibly empty sets
$\mathcal{C}_{\ell} \subseteq \mathcal{P}_{n-1}/K$ for
$\ell\in\{i,x,y,z\}$. In the above we use the convention that
$\sigma_\ell\otimes \mathcal{C} = \{\sigma_\ell \otimes P : P
\in\mathcal{C}\}$,
where we define $\sigma_{\ell}\otimes P$ to be the equivalence class
$[\sigma_{\ell}\otimes A] \in \mathcal{P}_n / K$ for any chosen
representative $A \in P$, the notion being well defined and
independent of the choice of the representative $A$. In many cases we
are not concerned with the exact labels of the sets and instead work
with the decomposition
\begin{equation}\label{Eq:DecomposeUVW}
\mathcal{S} =  (\sigma_i \otimes \mathcal{C}_i) \cup (\sigma_u \otimes
\mathcal{C}_u) \cup (\sigma_v \otimes \mathcal{C}_v) \cup (\sigma_w
\otimes \mathcal{C}_w),
\end{equation}
where $(u,v,w)$ is an arbitrary permutation of $(x,y,z)$ that
satisfies the condition that $\mathcal{C}_u = \emptyset$ implies
$\mathcal{C}_v = \emptyset$, and $\mathcal{C}_v = \emptyset$ implies
$\mathcal{C}_w = \emptyset$.

\section{Sets of commuting Paulis}\label{Sec:Commuting}

In this section we study the structure and cardinality of maximally
commuting sets of Paulis. One of the basic properties of these sets is
that they from subgroups, which stated in the following lemma.

\begin{lemma}\label{lem-subgroup-abelian-pauli}
If $\mathcal{S}\subseteq\mathcal{P}_n/K$ is maximally commuting,
then $\mathcal{S}$ is a subgroup of $\mathcal{P}_n/K$.
\end{lemma}
\begin{proof}
Since $I$ commutes with all elements in $\mathcal{P}_n/K$, it
follows that $I\in\mathcal{S}$ by maximality. If
$P, Q \in \mathcal{C}$ are distinct elements, then $PQ$ commutes
with all elements in $\mathcal{S}$, and therefore by maximality
$PQ \in \mathcal{C}$. Hence $\mathcal{S}$ is a subgroup of
$\mathcal{P}_n/K$.
\end{proof}

From this result it immediately follows using
Lemma~\ref{lem-generator-abelian-pauli1}(b), that
$\vert\mathcal{S}\vert$ is a power of two. From
Lemma~\ref{lem-generator-abelian-pauli2}, it further follows that
$\prod\mathcal{S} = I$ whenever $\ell \geq 2$. The next lemma
elaborates on the structure of maximally commuting subsets of abelian
Paulis.

\begin{lemma}{(Commuting structure lemma)}
\label{Lemma:GroupsCommuting}
Let $\mathcal{S} \subseteq \mathcal{P}_n / K$ be maximally commuting
with $n\geq 2$ and decomposition of the
form~\eqref{Eq:DecomposeUVW}. Then $I\in\mathcal{C}_i$, and the
following are true:

\begin{enumerate}[(a)]
\item For $\ell\in\{i,u,v,w\}$ the elements within
  $\mathcal{C}_{\ell}$ commute with each other, as well as with all
  elements in $\mathcal{C}_i$. The elements between any pair of sets
  $\mathcal{C}_u$, $\mathcal{C}_v$, and $\mathcal{C}_w$ anticommute.

\item If $\mathcal{C}_v = \mathcal{C}_w = \emptyset$, then
  $\mathcal{C}_i = \mathcal{C}_u$, and $\mathcal{C}_i$ is a maximally
  commuting set.

\item If $\mathcal{C}_u, \mathcal{C}_v \neq\emptyset$, the sets
  $\mathcal{C}_i$, $\mathcal{C}_u$, $\mathcal{C}_v$, and
  $\mathcal{C}_w$ satisfy the following properties:
\begin{enumerate}[1.]
\item for any $P\in\mathcal{C}_i$ we have $P*\mathcal{C}_i = \mathcal{C}_i*\mathcal{C}_i = \mathcal{C}_i$,
\item for any $P\in\mathcal{C}_u$ we have $P*\mathcal{C}_u = \mathcal{C}_u*\mathcal{C}_u = \mathcal{C}_i$,
\item for any $P\in\mathcal{C}_i$ and any $Q\in\mathcal{C}_u$ we have
  $P *\mathcal{C}_u = Q*\mathcal{C}_i = \mathcal{C}_i*\mathcal{C}_u =
  \mathcal{C}_u$,
\item for any $P\in\mathcal{C}_u$ and any $Q\in\mathcal{C}_v$ we have
  $P*\mathcal{C}_v = Q*\mathcal{C}_u = \mathcal{C}_u*\mathcal{C}_v =
  \mathcal{C}_w$,
\item
  $\vert \mathcal{C}_i\vert = \vert\mathcal{C}_u\vert =
  \vert\mathcal{C}_v\vert = \vert\mathcal{C}_w\vert$,
  and the sets are non-empty and disjoint
\item sets $\mathcal{C}_i$, $(\mathcal{C}_i \cup \mathcal{C}_u)$,
  $(\mathcal{C}_i \cup \mathcal{C}_v)$, and
  $(\mathcal{C}_i \cup \mathcal{C}_w)$ are subgroups of
  $\mathcal{P}_{n-1} / K$.
\end{enumerate}
\end{enumerate}
\end{lemma}

\begin{proof}
  If $I \not\in \mathcal{C}_i$ we can add it, so for maximally
  commuting sets we have $I\in\mathcal{C}_i$.

\textit{(a)} This follows directly from commutativity of the elements in $\mathcal{S}$.

\textit{(b)} If $\mathcal{C}_v$ and $\mathcal{C}_w$ are empty, we can
add any element of $\mathcal{C}_i$ to $\mathcal{C}_u$ and vice versa,
and for maximal sets they must therefore be equal. The set
$\mathcal{C}_i$ must also be maximally commuting, otherwise there
exists an $R \in (\mathcal{P}_{n-1} / K) \setminus \mathcal{C}_i$ that
commutes with $\mathcal{C}_i$ and could therefore be added to both
$\mathcal{C}_i$ and $\mathcal{C}_u$, which contradicts maximality of
$\mathcal{S}$.

\textit{(c)} Now assume that $\mathcal{C}_u$ and $\mathcal{C}_v$ are
non-empty.

\begin{enumerate}

\item Given any $P,Q \in \mathcal{C}_i$, $PQ$ commutes with all
  elements in the sets $\mathcal{C}_i$, $\mathcal{C}_u$,
  $\mathcal{C}_v$, and $\mathcal{C}_w$, and by maximality must
  therefore be an element of $\mathcal{C}_i$. For a fixed $P$, the
  products $PQ$ differ for all $Q \in \mathcal{C}_i$, and
  $P*\mathcal{C}_i$ must therefore coincide with $\mathcal{C}_i$.

\item Given any $P,Q\in\mathcal{C}_u$, it can be verified that $PQ$
  commutes with the elements in all the sets. By maximality we must
  therefore have that $PQ \in \mathcal{C}_i$ and
  $P*\mathcal{C}_u\subseteq\mathcal{C}_i$. Again for a fixed $P$, the
  products $PQ$ differ for all $Q \in \mathcal{C}_u$, and so
  $\vert P*\mathcal{C}_u\vert = \vert\mathcal{C}_u\vert \leq
  \vert\mathcal{C}_i\vert$.

\item Given $P\in\mathcal{C}_i$ and $Q\in\mathcal{C}_u$, $PQ$ commutes
  with all elements in $\mathcal{C}_i$ and $\mathcal{C}_u$, and
  anticommutes with all elements in $\mathcal{C}_v$ and
  $\mathcal{C}_w$. By maximality we must therefore have that
  $PQ \in \mathcal{C}_u$. Using arguments similar to the proof of
  properties 1--2, we thus conclude that
  $P*\mathcal{C}_u = \mathcal{C}_u$,
  $Q*\mathcal{C}_i \subseteq\mathcal{C}_u$, and
  $\vert Q*\mathcal{C}_i\vert = \vert\mathcal{C}_i\vert \leq
  \vert\mathcal{C}_u\vert$.

\item Given $P \in \mathcal{C}_u$ and $Q \in \mathcal{C}_v$, it can be
  verified that $PQ$ anticommutes with all elements in $\mathcal{C}_u$
  and $\mathcal{C}_v$, and commutes with all elements in
  $\mathcal{C}_i$ and $\mathcal{C}_w$. It again follows from
  maximality that $PQ \in \mathcal{C}_w$, and thus
  $P*\mathcal{C}_v \subseteq \mathcal{C}_w$,
  $Q*\mathcal{C}_u \subseteq\mathcal{C}_w$,
  $\vert P*\mathcal{C}_v \vert = \vert \mathcal{C}_v \vert$, and
  $\vert Q*\mathcal{C}_u \vert = \vert \mathcal{C}_u \vert$.

\item From the cardinality relations in the proof of properties 2--3
  it follows that $\vert\mathcal{C}_i\vert =
  \vert\mathcal{C}_u\vert$.
  Since the choice of $u$, $v$, and $w$ was arbitrary it follows that
  $\vert\mathcal{C}_u\vert = \vert\mathcal{C}_v\vert =
  \vert\mathcal{C}_w\vert$.
  Given that the cardinality of all four sets are equal it follows
  that the $\subseteq$ and $\leq$ relations in items 2--4 can be
  replaced with equality. As $\mathcal{C}_i$ contains $I$, we conclude
  that all the four subsets are non-empty. We know from (a) that all
  elements in $\mathcal{C}_w$ commute with $\mathcal{C}_i$ and
  anticommute with $\mathcal{C}_u$, and it must therefore hold that
  $\mathcal{C}_i\cap\mathcal{C}_u = \emptyset$. A similar argument
  applies to $\mathcal{C}_u \cap\mathcal{C}_v$, and all other pairs of
  sets.

\item Property $2$ shows that $\mathcal{C}_i$ is closed under the
  group operation of $\mathcal{P}_{n-1} / K$, while properties 2--4
  show that $\mathcal{C}_i \cup \mathcal{C}_u$,
  $\mathcal{C}_i \cup \mathcal{C}_v$, and
  $\mathcal{C}_i \cup \mathcal{C}_w$ are also closed under the same
  group operation.
\end{enumerate}\vspace*{-16pt}
\end{proof}

For the cardinality of maximally commuting sets in $\mathcal{P}_n/K$
we have the following result:

\begin{theorem}\label{Theorem:maximally-commuting-maximum}
A commuting set $\mathcal{S}\subseteq\mathcal{P}_n/K$ is maximally
commuting iff $\vert\mathcal{S}\vert = 2^n$.
\end{theorem}
\begin{proof}
Let $\mathcal{G}$ be a minimal generator set for
$\mathcal{S}$. Suppose by contradiction that
$k := \vert\mathcal{G}\vert> n$, then it follows from
Lemma~\ref{lem-pattern-generators} and the fact that elements
commute with themselves, that each commutativity map with
$\mathcal{G}$ is generated by $4^n / 2^k < 2^n$ elements. For all
$Q\in\langle\mathcal{G}\rangle$, the commutativity map with respect
to $\mathcal{G}$ is the all-commuting map, but this gives a
contradiction, since
$\vert\langle\mathcal{G}\rangle\vert = 2^k > 2^n$. Similarly,
suppose that $\vert\mathcal{G}\vert < n$. In this case it follows
from Lemma~\ref{lem-pattern-generators} that there must exist a
$P \in (\mathcal{P}_n/K) \setminus\langle\mathcal{G}\rangle$ that
commutes with all elements in $\mathcal{G}$, and therefore with all
elements in $\langle\mathcal{G}\rangle$. It follows that $P$ could
be added to $\mathcal{S}$, thus contradicting maximality.
\end{proof}

A slight strengthening of the commuting structure lemma is now possible.

\begin{corollary}
\label{cor-maximally-commuting-structure-strengthening}
Let $\mathcal{S}\subseteq\mathcal{P}_n/K$ be a maximally commuting set
with $n \geq 2$ and decomposition~\eqref{Eq:DecomposeUVW} with
$\mathcal{C}_w\neq\emptyset$. Then
$\vert \mathcal{C}_i\vert = \vert\mathcal{C}_u\vert =
\vert\mathcal{C}_v\vert = \vert\mathcal{C}_w\vert = 2^{n-2}$.
In addition, $(\mathcal{C}_i \cup \mathcal{C}_u)$,
$(\mathcal{C}_i \cup \mathcal{C}_v)$, and
$(\mathcal{C}_i \cup \mathcal{C}_w)$ are maximally commuting subgroups
of $\mathcal{P}_{n-1} / K$.
\end{corollary}
\begin{proof}
By Theorem \ref{Theorem:maximally-commuting-maximum},
$\vert \mathcal{S} \vert = 2^n$. Since each of the four sets
$\mathcal{C}$ have equal size by
Lemma~\ref{Lemma:GroupsCommuting}(c), it follows that each must have
size $2^n/4 = 2^{n-2}$. The set
$\mathcal{H} := \mathcal{C}_i\cup\mathcal{C}_{\ell}$ is commuting
for any $\ell\in\{u,v,w\}$. From property 5 of
Lemma~\ref{Lemma:GroupsCommuting}(c) we know that
$\mathcal{C}_i\cap\mathcal{C}_{\ell} = \emptyset$, and it therefore
follows that $\vert\mathcal{H}\vert = 2^{n-1}$, which is maximal by
Theorem~\ref{Theorem:maximally-commuting-maximum}.
\end{proof}

The next two lemmas provide a converse to Lemma \ref{Lemma:GroupsCommuting}.

\begin{lemma}
\label{lem-maximally-commuting-converse1}
Let $\mathcal{S} \subseteq \mathcal{P}_{n-1} / K$ be maximally
commuting. Then the set
$\mathcal{S}' = (\sigma_i \otimes \mathcal{S}) \; \cup \;
(\sigma_{\ell} \otimes \mathcal{S})$
is a maximally commuting subgroup of $\mathcal{P}_{n}/K$, for all
$\ell \in \{x,y,z\}$.
\end{lemma}

\begin{proof}
Without loss of generality, assume that $\ell = x$. Next, for the
sake of contradiction, suppose that $\mathcal{S}'$ is not maximally
commuting. As $\mathcal{S}$ is maximally commuting, there exists no
element of the form $\sigma_i \otimes R$, or $\sigma_x \otimes R$,
with $R \notin \mathcal{S}$, such that
$\mathcal{S}' \cup (\sigma_i \otimes R)$ is still commuting. So
there must exist an element of the form $\sigma_j \otimes R$,
$j \in \{y,z\}$, and $R \in \mathcal{P}_{n-1} / K$, such that
$\mathcal{S}' \cup \{\sigma_j \otimes R\}$ is mutually
commuting. But then $R$ must commute and anticommute simultaneously
with each element in $\mathcal{S}$, which is a contradiction.
\end{proof}

\begin{lemma}
\label{lem-maximally-commuting-converse2}
Suppose we have four subsets $\mathcal{C}_i$, $\mathcal{C}_x$,
$\mathcal{C}_y$, and $\mathcal{C}_z$ of $\mathcal{P}_{n-1} / K$, that
satisfy the property in Lemma \ref{Lemma:GroupsCommuting}(a), and also
satisfy
$\vert \mathcal{C}_i\vert = \vert\mathcal{C}_x\vert =
\vert\mathcal{C}_y\vert = \vert\mathcal{C}_z\vert = 2^{n-2}$.
Then the set
$\mathcal{S} = (\sigma_i \otimes \mathcal{C}_i) \cup (\sigma_u \otimes
\mathcal{C}_x) \cup (\sigma_v \otimes \mathcal{C}_y) \cup (\sigma_w
\otimes \mathcal{C}_z)$
is a maximally commuting subgroup of $\mathcal{P}_n / K$, for all
permutations $(u,v,w)$ of $(x,y,z)$. In particular this implies that
$\mathcal{C}_i$, $\mathcal{C}_x$, $\mathcal{C}_y$, and $\mathcal{C}_z$
also satisfy properties 1--6 of Lemma \ref{Lemma:GroupsCommuting}(c).
\end{lemma}

\begin{proof}
It is easily checked that $\mathcal{S} \subseteq \mathcal{P}_n / K$
is a commuting set. We also have that $|\mathcal{S}| = 2^n$, and
hence by Theorem \ref{Theorem:maximally-commuting-maximum} it is a
maximally commuting subgroup. Hence the sets $\mathcal{C}_i$,
$\mathcal{C}_x$, $\mathcal{C}_y$, and $\mathcal{C}_z$ satisfy all
the properties 1--6 of Lemma \ref{Lemma:GroupsCommuting}(c).
\end{proof}


\section{Sets of anticommuting Paulis}\label{Sec:Anticommuting}

In this section we study the structure of sets of maximally
anticommuting abelian Paulis. After clarifying the basic structure in
Section~\ref{Sec:AnticommutingBasics}, we consider possible sizes of
these sets and properties of sets that attain the maximum size in
Sections~\ref{Sec:AnticommutingSize}
and~\ref{Sec:AnticommutingMaximum}. We provide an efficient algorithm
for creating various types of maximally anticommuting sets in
Section~\ref{Sec:AnticommutingGeneration}.

\subsection{Structure of maximally anticommuting sets}\label{Sec:AnticommutingBasics}

We start with a number of basic facts on sets of anticommuting abelian Paulis. 

\begin{theorem}\label{Thm:AnticommutingBasics}
Let $\mathcal{G} = \{P_1,\ldots,P_k\}$ be a set of anticommuting Paulis, then

\begin{enumerate}[(a)]
\setmargins
\item if $k$ is even, then $Q = \prod\mathcal{G}$ anticommutes with $\mathcal{G}$,
and $\mathcal{G}\cup \{Q\}$ is maximally anticommuting,\label{Thm:AnticommutingBasics-Q}

\item $\mathcal{G}$ is maximal implies that $k$ is
  odd,\label{Thm:AnticommutingBasics-n-odd}
\item $\prod\mathcal{G} = I$ implies that $\mathcal{G}$ is maximal and
  $k$ is odd, \label{Thm:AnticommutingBasics-prod-I}
\item for any proper non-empty subset
  $\mathcal{H} \subset \mathcal{G}$ it holds that
  $\prod\mathcal{H} \neq I$,\label{Thm:AnticommutingBasics-R}
\item $\prod \mathcal{G} \neq I$ implies that $\mathcal{G}$ is a
  minimal generating set for a subgroup of order
  $2^k$, \label{Thm:AnticommutingBasics-minimal-generating}
\item $\prod \mathcal{G} = I$ implies that $\mathcal{G}$ is a
  generating set for a subgroup of order
  $2^{k-1}$. \label{Thm:AnticommutingBasics-minimal-generating-prod-I}
\end{enumerate}
\end{theorem}

\begin{proof}
\textit{(\ref{Thm:AnticommutingBasics-Q})} To determine
commutativity we can take any Pauli operator represented by the
equivalence classes, which we shall indicate by a bar. For any
$i\in[k]$ it takes $k$ pairwise matrix swaps to convert
$\bar{P}_i\bar{Q}$ to $\bar{Q}\bar{P}_i$. Each swap with a term
other than $\bar{P}_i$ leads to a multiplication by $-1$ due to
anticommutativity. Given that only one of the $k$ swaps is with
$\bar{P}_i$ itself, it follows from $k-1$ is odd, that
$\bar{P}_i\bar{Q} = -\bar{Q}\bar{P}_i$, and therefore that
$\bar{P}_i$ and $\bar{Q}$ anticommute. The result that $Q$
anticommutes with all $\mathcal{G}$ follows by observing that both
$i$ and the operator representation was arbitrary. Now, suppose
there exists a
$P \in (\mathcal{P}_n/K \setminus \langle\mathcal{G}\rangle)$ that
anticommutes with $\mathcal{G}$, then any corresponding operator
$\bar{P}$ anticommutes with all even $k$ terms that comprise
$\bar{Q}$, and therefore commutes with $Q$. It follows that
$\mathcal{G}\cup\{Q\}$ cannot be extended, and is therefore
maximally anticommuting.

\textit{(\ref{Thm:AnticommutingBasics-n-odd})} It follows from (a),
that if $k$ is even we can add $Q$ to $\mathcal{G}$, which means
that $\mathcal{G}$ is not maximal.

\textit{(\ref{Thm:AnticommutingBasics-prod-I})} If
$\mathcal{G}=\{I\}$ then the result is clear, and we therefore
assume that $k \geq 2$. Suppose $k$ is even, then we have from (a)
that $\prod \mathcal{G}$ anticommutes with every element in
$\mathcal{G}$, the product could therefore not have been $I$, and it
follows that $k$ must be odd. Define
$\mathcal{G}'=\{P_1,\ldots,P_{k-1}\}$, then
$P_{k} = Q = \prod\mathcal{G}'$. Applying (a) to $\mathcal{G}'$ then
shows that $\mathcal{G} = \mathcal{G}'\cup\{Q\}$ is maximally
anticommuting.

\textit{(\ref{Thm:AnticommutingBasics-R})} Suppose that
$\prod\mathcal{H}=I$. Then from the previous result it follows that
$\mathcal{H}$ is maximal. Consequently,
$\mathcal{G}\setminus\mathcal{H}$ cannot anticommute with
$\mathcal{H}$, thereby contradicting the fact that $\mathcal{G}$ is
anticommuting.

\textit{(\ref{Thm:AnticommutingBasics-minimal-generating})} For sake
of contradiction suppose that $\mathcal{G}$ is not a minimal
generating set (note that $I \not\in \mathcal{G}$ since it is
anticommuting). Then there exists $i\in[k]$, such that
$\langle \mathcal{G} \rangle = \langle \mathcal{G} \setminus \{P_i\}
\rangle$.
But this implies that $P_i = \prod \mathcal{H}$ for some
$\mathcal{H} \subseteq \mathcal{G} \setminus \{P_i\}$, and thus
$P_i (\prod \mathcal{H}) = I$. If
$\mathcal{H} \subset \mathcal{G} \setminus \{P_i\}$ is a proper
subset, then we get a contradiction by
(\ref{Thm:AnticommutingBasics-R}), while if
$\mathcal{H} = \mathcal{G} \setminus \{P_i\}$ we also get a
contradiction because that would imply $\prod \mathcal{G} =
I$.
Finally $|\langle \mathcal{S} \rangle| = 2^k$ by Lemma
\ref{lem-generator-abelian-pauli1}(b).

\textit{(\ref{Thm:AnticommutingBasics-minimal-generating-prod-I})}
If $\mathcal{G} = \{I\}$ the statement is true. Otherwise
$I \notin \mathcal{G}$, and $|\mathcal{G}| \geq 2$. In this case
$I \neq P_1 = \prod (\mathcal{G} \setminus \{P_1\})$, and so
$\langle \mathcal{G} \rangle = \langle \mathcal{G} \setminus \{P_1\}
\rangle$,
and by (\ref{Thm:AnticommutingBasics-minimal-generating}),
$\mathcal{G} \setminus \{P_1\}$ is a minimal generating set for the
subgroup $\langle \mathcal{G} \setminus \{P_1\} \rangle$ of order
$2^{k-1}$.
\end{proof}

The above result also shows that if $\mathcal{G} \neq \{I\}$ is an
anticommuting set with $\prod\mathcal{G}=I$, then we can create a
minimum generating set by removing any single element. Next we prove
an important structure theorem for any maximally anticommuting subset
of $\mathcal{P}_n / K$.

\begin{theorem}{(Anticommuting structure theorem)}
\label{thm-anticommuting-structure}
Let $\mathcal{G} \subseteq \mathcal{P}_n / K$ be maximally
anticommuting with decomposition~\eqref{Eq:DecomposeUVW}. Then the
following statements are true.

\begin{enumerate}[(i)]
\setmargins
\item The elements within each of the sets anticommute, and elements
  in $\mathcal{C}_i$ anticommute with $\mathcal{C}_\ell$ for
  $\ell\in\{u,v,w\}$. Elements between $\mathcal{C}_u$,
  $\mathcal{C}_v$, and $\mathcal{C}_w$ commute.

\item Decomposition \eqref{Eq:DecomposeUVW} has exactly one of the
  following forms:\vspace*{3pt}
\begin{center}
\begin{tabular}{llp{8.0cm}}
\hline
\multicolumn{2}{l}{Non-empty sets} & Properties\\
\hline
(a)& \!\!\!$\mathcal{C}_i$ & $\mathcal{C}_i$ is maximally anticommuting and $\vert\mathcal{G}\vert < 2n$.
\\
(b) & \!\!\!$\mathcal{C}_i$, $\mathcal{C}_u$ & $\vert\mathcal{C}_i\vert$ is odd and $\vert\mathcal{C}_u\vert$ is even, $\mathcal{C}_i\cup \mathcal{C}_u$ is maximally anticommuting,
$\vert\mathcal{G}\vert < 2n$.
\\
(c) & \!\!\!$\mathcal{C}_i$, $\mathcal{C}_u$, $\mathcal{C}_v$ & $\vert\mathcal{C}_i\vert$ is odd and $\vert\mathcal{C}_u\vert$ and $\vert\mathcal{C}_v\vert$ are even.
\\
(d) & \!\!\!$\mathcal{C}_u$, $\mathcal{C}_v$, $\mathcal{C}_w$ & $\vert\mathcal{C}_\ell\vert$ is odd for all $\ell\in\{u,v,w\}$.
\\
(e) & \!\!\!all & $\vert\mathcal{C}_{u}\vert$, $\vert\mathcal{C}_v\vert$, and $\vert\mathcal{C}_w\vert$ are either all odd (even), and $\vert\mathcal{C}_i\vert$ is even (odd).\\
\hline
\end{tabular}
\end{center}

\item The sets $\mathcal{C}_i$, $\mathcal{C}_\ell$ are disjoint and
  $\vert\mathcal{C}_i\cup\mathcal{C}_\ell\vert$ is odd for all
  $\ell \in \{u,v,w\}$. The sets $\mathcal{C}_a$, $\mathcal{C}_b$ are
  disjoint whenever $|\mathcal{C}_a| > 1$ or $|\mathcal{C}_b| > 1$,
  for every distinct $a, b \in \{u,v,w\}$.
\end{enumerate}
\end{theorem}

\begin{proof}
\textit{(i)} The commutativity relations are easily verified.

\textit{(ii)} The decomposition has exactly one of the given forms
(a)--(e). By contradiction, if only $\mathcal{C}_u$ is non-empty we
can add $\sigma_z \otimes \sigma_i$. Likewise, if only $\mathcal{C}_u$
and $\mathcal{C}_v$ are non-empty, then one of the sets, say $\ell$,
has even size. It follows that we can add $\sigma_\ell \otimes P$ with
$P=\prod\mathcal{C}_\ell$.

In cases (a) and (b) note that we can omit the first element of all
Paulis without affecting the commutativity. It follows that
$\mathcal{C}_i \cup \mathcal{C}_u$ is a set of maximally anticommuting
$(n-1)$-Paulis, and we must therefore have that
$\vert \mathcal{G} \vert \leq 2(n-1)+1 < 2n$.

In cases (b) and (c), suppose that $\vert\mathcal{C}_i\vert$ is
even. Then $P=\prod\mathcal{C}_i$ anticommutes with $\mathcal{C}_i$
and commutes with the other sets, and we can therefore add
$\sigma_w\otimes P$ to $\mathcal{G}$. It follows that
$\vert\mathcal{C}_i\vert$ must be odd. For (c), we show that
$\vert\mathcal{C}_u\vert$ must be even. Since $\vert\mathcal{G}\vert$
is odd, it follows that $\vert\mathcal{C}_v\vert$ is also
even. Suppose by contradiction that $\vert\mathcal{C}_u\vert$ is
odd. Then $P=\prod(\mathcal{C}_i\cup\mathcal{C}_u)$ anticommutes with
$\mathcal{C}_i$, $\mathcal{C}_u$, and $\mathcal{C}_v$. It follows that
we can add $\sigma_i\otimes P$.

In case (d), suppose that, without loss of generality,
$\vert\mathcal{C}_u\vert$ is even. Then $P=\prod\mathcal{C}_u$
anticommutes with $\mathcal{C}_u$, but commutes with both
$\mathcal{C}_v$ and $\mathcal{C}_w$, since $\mathcal{C}_v$ and
$\mathcal{C}_w$ commute with all matrices in $\mathcal{C}_u$. It
follows that we can add $\sigma_u\otimes P$, which contradicts
maximality. Since the choice of $\mathcal{C}_u$ was arbitrary it
follows that all sets must have odd cardinality.

In case (e), suppose that
$\vert\mathcal{C}_i\vert + \vert\mathcal{C}_u\vert$ is even. Then we
can form $P=\prod(\mathcal{C}_i\cup\mathcal{C}_u)$, which anticommutes
with $\mathcal{C}_i$ and $\mathcal{C}_u$ and either commutes or
anticommutes with both $\mathcal{C}_v$ and $\mathcal{C}_w$. If it
commutes with both we can add $\sigma_u\otimes P$, otherwise we can
add $\sigma_i\otimes P$. It follows that
$\vert\mathcal{C}_i\vert + \vert\mathcal{C}_\ell\vert$ is odd for all
$\ell\in\{u,v,w\}$. If $\vert\mathcal{C}_i\vert$ is even, then
$\vert\mathcal{C}_{u}\vert$, $\vert\mathcal{C}_v\vert$, and
$\vert\mathcal{C}_w\vert$ are all odd, and vice versa.

\textit{(iii)} First suppose that
$P \in \mathcal{C}_i \cap \mathcal{C}_\ell \neq \emptyset$, for
$\ell \in \{u,v,w\}$. Then $\sigma_i \otimes P$ and
$\sigma_\ell \otimes P$ are both in $\mathcal{G}$ but commute, which
is a contradiction. Now suppose, without loss of generality, that
$|\mathcal{C}_u| > 1$, and
$Q \in \mathcal{C}_u \cap \mathcal{C}_v \neq \emptyset$. Then there
exists $R \in \mathcal{C}_u$, different from $Q$, which anticommutes
with $Q$. But then $\sigma_u \otimes R$ commutes with
$\sigma_v \otimes Q$, which is again a contradiction.  The fact that
$\vert\mathcal{C}_i\cup\mathcal{C}_{\ell}\vert$ is odd follows
directly from the decompositions in (ii).
\end{proof}

The following observation now follows, which is the most important result of this section.

\begin{corollary}
\label{lem-anticommuting-prod}
An anticommuting subset $\mathcal{G} \subseteq \mathcal{P}_n / K$ is
maximally anticommuting iff $\prod\mathcal{G} = I$.
\end{corollary}

\begin{proof}
The ``if'' part was already proved in
Theorem~\ref{Thm:AnticommutingBasics}(\ref{Thm:AnticommutingBasics-prod-I}). For
the other direction, assume that
$\mathcal{G} \subseteq \mathcal{P}_n / K$ is maximally
anticommuting. Without loss of generality, choose any term index of
the underlying Pauli operators and permute the term order such that
the selected index is the first one. It suffices to show that the
product of all the elements in $\mathcal{G}$ can be written as
$\sigma_i \otimes P$, since the result then holds for all terms due
to the fact that the selected index was arbitrary. To complete the
proof, consider the decomposition in
\eqref{Eq:DecomposeUVW}. Theorem~\ref{thm-anticommuting-structure}(ii)
guarantees that only one of the cases (a)--(e) can occur, and in
each case the product of the first term is $\sigma_i$, as desired.
\end{proof}

Another interesting corollary is the following:
\begin{corollary}\label{cor-product-three-replace}
Let $\mathcal{G} \subseteq \mathcal{P}_n / K$ be maximally
anticommuting, and suppose $|\mathcal{G}| = 2m+1 \geq 3$. Define
$\mathcal{H}_m = \mathcal{G}$, and for $0 < k \leq m$, recursively
define
$\mathcal{H}_{k-1} = (\mathcal{H}_{k} \setminus \mathcal{J}_{k})
\cup (\prod \mathcal{J}_{k})$,
for any $\mathcal{J}_k \subseteq \mathcal{H}_{k}$ with
$|\mathcal{J}_k| = 3$. Then $\mathcal{H}_k$ is maximally
anticommuting for all $0 \leq k \leq m$, and
$\mathcal{H}_0 = \{I\}$.
\end{corollary}
\begin{proof}
We proceed by induction. $\mathcal{H}_m$ is maximally anticommuting
by definition. Now assume that $\mathcal{H}_k$ is maximally
anticommuting for some positive $k \leq m$. Then by
Corollary~\ref{lem-anticommuting-prod}, we have
$\prod \mathcal{H}_k = I$, and so by construction
$\prod \mathcal{H}_{k-1} = \prod \mathcal{H}_k = I$. Moreover
$\mathcal{H}_{k-1}$ is anticommuting, and using
Corollary~\ref{lem-anticommuting-prod} again, we conclude that
$\mathcal{H}_{k-1}$ is maximally anticommuting. This completes the
induction step. Note that the final set satisfies
$\mathcal{H}_0 = \{\prod\mathcal{H}_{n}\} = \{I\}$.
\end{proof}

\subsection{Size of maximally anticommuting sets}\label{Sec:AnticommutingSize}

For $1$-Paulis we find that $\{\sigma_i\}$ and
$\{\sigma_x,\sigma_y,\sigma_z\}$ are maximally anticommuting sets. We
can hierarchically generate sets of higher-dimensional anticommuting
matrices from existing sets. For example, given a set $\mathcal{G}_n$
of maximally anticommuting $n$-Paulis, we can
\begin{enumerate}
\item Set $\mathcal{G}_{n+1} = (\sigma_x \otimes \mathcal{G}_{n}) \cup
  (\sigma_y \otimes I) \cup (\sigma_z \otimes I)$.
\item Set
  $\mathcal{G}_{2n+1} = (\sigma_x \otimes I\otimes \mathcal{G}_{n})
  \cup (\sigma_y \otimes \mathcal{G}_{n} \otimes I)\cup (\sigma_z
  \otimes I \otimes I)$.
\item Given an odd number of maximal sets $\mathcal{S}_i$ of equal
  cardinality, then taking Kronecker products of corresponding
  elements in these sets gives a maximally anticommuting set, since
  the new elements anticommute and multiply to $I$.
\end{enumerate}

\noindent
Repeated application of the first construction with the initial set
$\mathcal{G}_1 = \{\sigma_x,\sigma_y,\sigma_z\}$, gives a set
$\mathcal{G}_n$ with $\vert\mathcal{G}_n\vert = 2n+1$. The following
results clarify possible sizes of maximally anticommuting sets, and
show that the above $\mathcal{G}_n$ attains the maximum size for sets
of anticommuting elements in $\mathcal{P}_n/K$.

\begin{lemma}\label{lem-max-anticommuting-paulis}
If $\mathcal{G} \subseteq \mathcal{P}_n / K$ is anticommuting, then $\vert\mathcal{G}\vert \leq 2n+1$.
\end{lemma}
\begin{proof}
  We note that this lemma is well-known; it follows for example from
  Proposition 9 in \cite{hrubes2016families}, and is also proved in
  ~\cite{bonet2019nearly}. Here we give an elegant and simpler proof
  of this fact. For the sake of contradiction, suppose
  $|\mathcal{G}| > 2n+1$. If $\prod \mathcal{G} \neq I$, then
  $\mathcal{G}$ is a minimal generating set, while if
  $\prod \mathcal{G} = I$, we can exclude any element
  $P \in \mathcal{G}$ and then $\mathcal{G} \setminus \{P\}$ is a
  minimal generating set by Theorem
  \ref{Thm:AnticommutingBasics}(\ref{Thm:AnticommutingBasics-minimal-generating}). In
  either case we have a minimal generating set of cardinality at least
  $2n+1$. By Lemma \ref{lem-generator-abelian-pauli1}(b), the set
  generates a subgroup of $\mathcal{P}_n / K$ of order at least
  $2^{2n+1} > 4^n$, which is a contradiction.
\end{proof}

\begin{corollary}\label{cor-max-anticommuting-paulis-exist}
  For every odd integer $\ell$ up to and including $2n+1$, there
  exists a maximally anticommuting subset of $\mathcal{P}_n / K$ of
  cardinality $\ell$.
\end{corollary}
\begin{proof}
  We know from the example at the beginning of this section that
  maximally anticommuting subsets of size $2n+1$ exist in
  $\mathcal{P}_n/K$, so take any such set $\mathcal{G}$. The result
  then follows by applying the construction in
  Corollary~\ref{cor-product-three-replace}.
\end{proof}

\subsection{Anticommuting sets of maximum size}\label{Sec:AnticommutingMaximum}

In the next theorem, we clarify the structure of maximally
anticommuting subsets of $\mathcal{P}_n / K$ that attain the maximum
size.

\begin{theorem}\label{thm-maximum-anticommuting}
Given a maximally anticommuting set $\mathcal{G} \subseteq \mathcal{P}_n/K$ of size $2n+1$, with decomposition~\eqref{Eq:DecomposeXYZ}. Then
\begin{enumerate}[(a)]
\setmargins
\item $\prod (\mathcal{C}_i\cup\mathcal{C}_{\ell}) = I$ for $\ell\in\{x,y,z\}$.
\item $\mathcal{C}_i\cup\mathcal{C}_{\ell}$ is a maximally anticommuting set for $\ell\in\{x,y,z\}$.
\item Sets $\mathcal{C}_x$, $\mathcal{C}_y$, and $\mathcal{C}_z$ are non-empty.
\item $\prod \mathcal{C}_i = \prod \mathcal{C}_x = \prod \mathcal{C}_y = \prod \mathcal{C}_z$. Moreover, $\prod \mathcal{C}_x = \prod \mathcal{C}_y = \prod \mathcal{C}_z = I$ iff $\mathcal{C}_i = \emptyset$. 
\item $\prod\mathcal{C}_{\ell}\prod\mathcal{C}_{m} = I$ for all $\ell,m\in\{i,x,y,z\}$.
\end{enumerate}

\end{theorem}
\begin{proof}
\textit{(a)} Let $(u,v,w)$ be an arbitrary permutation of $(x,y,z)$.
By Theorem~\ref{thm-anticommuting-structure}(i), the set
$\mathcal{T} = \mathcal{C}_i\cup\mathcal{C}_u$ is anticommuting in
$P_{n-1}/K$ and we therefore have $\vert\mathcal{T}\vert \leq
2n-1$.
If this holds with equality we have an anticommuting set of maximum
size, which implies $\prod\mathcal{T} = I$. Otherwise we have
$\vert\mathcal{C}_v\vert+\vert\mathcal{C}_w\vert \geq 3$.  Suppose
that $\prod\mathcal{T} \neq I$, then $\mathcal{T}$ forms a minimal
generator for a subgroup of $P_{n-1}/K$. Since $\mathcal{G}$
generates $P_n/K$, it follows that we can find
$\mathcal{D} \subseteq (\mathcal{C}_v\cup\mathcal{C}_w)$ with
$\vert\mathcal{D}\vert=3$, such that
$\mathcal{G}' := \mathcal{T} \cup
((\mathcal{C}_v\cup\mathcal{C}_w)\setminus\mathcal{D})$
generates $\mathcal{P}_{n-1}/K$. Because $\vert\mathcal{D}\vert =3$,
it follows that $\vert\mathcal{D}\cap\mathcal{C}_{\ell}\vert \geq 2$
for exactly one $\ell\in\{v,w\}$, and we can therefore find distinct
elements $P,Q \in (\mathcal{D}\cap\mathcal{C}_\ell)$. Both elements
are generated by $\mathcal{G}'$ and have the same commutativity map
with the elements in $\mathcal{G}'$. From
Lemma~\ref{lem-pattern-generators} we know that each commutativity
map with respect to $\mathcal{G}'$ occurs exactly once, which give a
contradiction. It must therefore hold that $\prod\mathcal{T}=I$, and
the result follows by noting that the choice of $u$ was arbitrary.

\textit{(b)} This follows from (a) because
$\mathcal{C}_i\cup\mathcal{C}_{\ell}$ is an anticommuting set by
Theorem \ref{thm-anticommuting-structure}(i).

\textit{(c)} Consider the
decomposition~\eqref{Eq:DecomposeUVW}. Then from
Theorem~\ref{thm-anticommuting-structure}(ii) we have that
$\mathcal{C}_u$ and $\mathcal{C}_v$ must be non-empty, otherwise
$\vert\mathcal{G}\vert < 2n$. By maximality of $\mathcal{G}$ it
follows from Corollary~\ref{lem-anticommuting-prod} that
$\prod\mathcal{G} = I$. Now suppose $\mathcal{C}_w = \emptyset$,
then
$\prod(\mathcal{C}_i\cup\mathcal{C}_u)\cdot\prod\mathcal{C}_v =
I$.
Using the fact $\prod(\mathcal{C}_i\cup\mathcal{C}_u) = I$ for (a),
we conclude that $\prod\mathcal{C}_v = I$, which means that
$\mathcal{C}_v$ is maximally anticommuting. From maximality of
$\mathcal{C}_i\cup\mathcal{C}_v$, we must therefore have that
$\mathcal{C}_i = \emptyset$. However,
Theorem~\ref{thm-anticommuting-structure}(ii) shows that maximally
anticommuting sets cannot have
$\mathcal{C}_i=\mathcal{C}_w=\emptyset$. This contradiction shows
that $\mathcal{C}_w$ cannot be empty, and consequently all three
sets $\mathcal{C}_x$, $\mathcal{C}_y$, and $\mathcal{C}_z$ are
non-empty.

\textit{(d)} If $\mathcal{C}_i \neq \emptyset$, then combining with
(c) we have that all four sets
$\mathcal{C}_i, \mathcal{C}_x, \mathcal{C}_y, \mathcal{C}_z$ are
non-empty. The result then again follows from (a) because it implies
that $\prod \mathcal{C}_i = \prod \mathcal{C}_\ell$ for
$\ell\in\{x,y,z\}$. If $\mathcal{C}_i = \emptyset$, the result is
true from (a) by substituting $\mathcal{C}_i =
\emptyset$.
Conversely, if $\prod\mathcal{C}_{\ell}=I$ for any
$\ell\in\{x,y,z\}$, then $\mathcal{C}_{\ell}$ is maximally
anticommuting, and because $\mathcal{C}_i\cup\mathcal{C}_{\ell}$ is
also maximally anticommuting by (b),
Theorem~\ref{thm-anticommuting-structure}(iii) then implies that
$\mathcal{C}_i=\emptyset$.

\textit{(e)} This follows from (d).
\end{proof}

\subsection{Extending an anticommuting set to its maximum size}\label{Sec:AnticommutingGeneration}

Given an anticommuting set $\mathcal{S} \subseteq \mathcal{P}_n / K$,
an interesting question is whether it can be extended to the maximum
possible size of $2n + 1$. By definition this cannot be done if
$\mathcal{S}$ is maximally anticommuting, or when
$|\mathcal{S}| = 2n+1$, in which case there is nothing to do. So the
interesting case is when $\prod \mathcal{S} \neq I$. We show that this
is possible in the following lemma.

\begin{lemma}
\label{lem-extend-anticommuting-sets}
Let $\mathcal{S} \subseteq \mathcal{P}_n / K$ be an anticommuting set
that is not maximally anticommuting. Then $\mathcal{S}$ can be
extended to a maximally anticommuting set of cardinality $2n+1$.
\end{lemma}

\begin{proof}
As $\mathcal{S}$ is not maximally anticommuting,
$\mathcal{S} \neq \{I\}$ and $|\mathcal{S}| < 2n+1$. We note that
the lemma is proved if we can show that if $|\mathcal{S}| < 2n$,
then $\mathcal{S}$ can be extended to an anticommuting set of
cardinality $|\mathcal{S}| + 1$ that is not maximally
anticommuting. Repeating this process we can extend $\mathcal{S}$ to
a set of cardinality $2n$, after which the construction in Theorem
\ref{Thm:AnticommutingBasics}(\ref{Thm:AnticommutingBasics-Q}) gives
the desired result.

If $|\mathcal{S}|$ is odd, there exists an element
$P \notin \mathcal{S}$, such that $\mathcal{S} \cup \{P\}$ is
anticommuting, and because $|\mathcal{S} \cup \{P\}|$ is even, by
Theorem \ref{Thm:AnticommutingBasics}
(\ref{Thm:AnticommutingBasics-n-odd}) $\mathcal{S} \cup \{P\}$ is
not maximally anticommuting. Now assume that $|\mathcal{S}|$ is
even, and $|\mathcal{S}| < 2n$. Pick an element
$Q \notin \langle \mathcal{S} \rangle$, and partition
$\mathcal{S} = \mathcal{C} \sqcup \mathcal{A}$ (one of the sets
possibly empty), such that $Q$ commutes with all elements in
$\mathcal{C}$, and $Q$ anticommutes with all elements in
$\mathcal{A}$. Choosing $R$ as
\begin{equation}
R = 
\begin{cases}
Q (\prod \mathcal{A}) & \;\; \text{if } |\mathcal{C}| \text{ is odd}, \\
Q (\prod \mathcal{C}) & \;\; \text{if } |\mathcal{C}| \text{ is even}, \; \mathcal{C} \neq \emptyset, \\
Q & \;\;  \text{if } \mathcal{C} = \emptyset, \\
\end{cases}
\end{equation}
we find that $R$ anticommutes with all elements in $\mathcal{S}$, and
so $\mathcal{S} \cup \{R\}$ is anticommuting. Moreover
$R (\prod \mathcal{S}) \neq I$, because otherwise we would have
$Q \in \langle \mathcal{S} \rangle$. This implies that
$\mathcal{S} \cup \{R\}$ is not maximally anticommuting, which
finishes the proof.
\end{proof}

Lemma \ref{lem-extend-anticommuting-sets} raises some interesting questions:
\begin{enumerate}[\itshape(1)]
\item \textit{Given an anticommuting set
    $\mathcal{S} \subseteq \mathcal{P}_n / K$ that is not maximally
    anticommuting, in how many distinct ways can we extend it to a
    bigger size?}

\item \textit{Is there an efficient algorithm to perform the
    extension?}
\end{enumerate}

In order to answer the above questions, we need to develop a better
understanding of the cosets of $\langle \mathcal{S} \rangle$. The
first question is answered in Section
\ref{Sec:anticommuting-counting}. For the second question, it turns
out that there exists an efficient randomized algorithm to extend
$\mathcal{S}$ to a bigger anticommuting set. We start by
characterizing a simple function that is important in the proof of the
following theorem.

\begin{lemma}
\label{lem-set-function-prop}
Let $\mathcal{S}$ be a set with $|\mathcal{S}| = m \geq 1$, and let $v : \mathcal{S} \rightarrow \{1,-1\}$ be any arbitrary map. Define a map $F_v : 2^\mathcal{S} \times \mathcal{S} \rightarrow \{1,-1\}$ by
\begin{equation}
\label{eq:set-function}
F_v (\mathcal{T}, x) = 
\begin{cases}
v(x)(-1)^{|\mathcal{T}| - 1} & \;\; \text{if } x \in \mathcal{T}, \\
v(x)(-1)^{|\mathcal{T}|} & \;\; \text{if } x \notin \mathcal{T},
\end{cases}
\end{equation}

and also define
\begin{equation}
f(v) := \left( \sum \limits_{x \in \mathcal{S}} (1 + v(x)) / 2 \right) = \sum \limits_{x \in \mathcal{S}} \mathbb{1}_{\{v(x) = 1\}}.
\label{eq:commuting-count-notation}
\end{equation}

If $m$ is even, then for every map
$q: \mathcal{S} \rightarrow \{1,-1\}$, there exists a unique subset
$\mathcal{U} \subseteq S$ (possibly empty), such that
$F_v (\mathcal{U}, \cdot) = q$. If $m$ is odd, then we have the
following cases:

\begin{enumerate}[(a)]
\setmargins
\item If $q: \mathcal{S} \rightarrow \{1,-1\}$ is a map such that
  $f(q) \equiv f(v) \mod 2$, then there exist exactly two subsets
  $\mathcal{V}, \; \mathcal{S} \setminus \mathcal{V} \in
  2^{\mathcal{S}}$
  (possibly with one empty), such that
  $F_v (\mathcal{V}, \cdot) = F_v (\mathcal{S} \setminus \mathcal{V},
  \cdot) = q$.

\item If $q: \mathcal{S} \rightarrow \{1,-1\}$ is a map such that
  $f(q) \not\equiv f(v) \mod 2$, then there does not exist any subset
  $\mathcal{V}$ of $\mathcal{S}$ such that
  $F_v (\mathcal{V}, \cdot) = q$.
\end{enumerate}
\end{lemma}

\begin{proof}
We need a fact that for arbitrary subsets
$\mathcal{A}, \mathcal{B} \subseteq \mathcal{S}$, the function in
\eqref{eq:set-function} satisfies
\begin{equation}
F_w (\mathcal{B}, \cdot) = F_v ((\mathcal{A} \cup \mathcal{B}) \setminus (\mathcal{A} \cap \mathcal{B}), \cdot), \text{ where } w = F_v (\mathcal{A}, \cdot).
\label{eq:set-function-identity}
\end{equation}
This is true because
\begin{equation}
F_w (\mathcal{B}, x) = 
\begin{cases}
v(x) (-1)^{|\mathcal{A}| + |\mathcal{B}|} & \;\; \text{ if } x \in (\mathcal{A} \cap \mathcal{B}) \sqcup (\mathcal{S} \setminus (\mathcal{A} \cup \mathcal{B})), \\
v(x) (-1)^{|\mathcal{A}| + |\mathcal{B}| - 1} & \;\; \text{ if } x \in (\mathcal{A} \setminus \mathcal{B}) \sqcup (\mathcal{B} \setminus \mathcal{A}),
\end{cases}
\end{equation}
from the definitions, and moreover
$|\mathcal{A}| + |\mathcal{B}| = |\mathcal{A} \setminus \mathcal{B}| +
|\mathcal{B} \setminus \mathcal{A}| + 2 |\mathcal{A} \cap \mathcal{B}|
= |(\mathcal{A} \cup \mathcal{B}) \setminus (\mathcal{A} \cap
\mathcal{B})| + 2 |\mathcal{A} \cap \mathcal{B}|$.

We first assume that $m$ is even. Let
$\mathcal{S} = \mathcal{S}_1 \sqcup \mathcal{S}_2$, such that
$v(x) = 1$ for all $x \in \mathcal{S}_1$, and $v(x) = -1$ for all
$x \in \mathcal{S}_2$. If we define
$\mathcal{T} \subseteq \mathcal{S}$ by
\begin{equation}
\mathcal{T} = 
\begin{cases}
\mathcal{S}_1 \;\; & \text{ if } |\mathcal{S}_1| \text{ is odd}, \\
\mathcal{S}_2 \;\; & \text{ if } |\mathcal{S}_1| \text{ is even},
\end{cases}
\end{equation}
then $F_v (\mathcal{T},x) = 1$ for all $x \in \mathcal{S}$. For any
map $q : \mathcal{S} \rightarrow \{1,-1\}$, let
$\mathcal{S} = \mathcal{T}_1 \sqcup \mathcal{T}_2$, such that
$q(x) = 1$ for all $x \in \mathcal{T}_1$, and $q(x) = -1$ for all
$x \in \mathcal{T}_2$. Then if we define
$\mathcal{R} \subseteq \mathcal{S}$ by
\begin{equation}
\mathcal{R} = 
\begin{cases}
\mathcal{T}_1 \;\; & \text{ if } |\mathcal{T}_1| \text{ is odd}, \\
\mathcal{T}_2 \;\; & \text{ if } |\mathcal{T}_1| \text{ is even},
\end{cases}
\end{equation}
we have that $F_{F_v (\mathcal{T},\cdot)} (\mathcal{R},\cdot) = q$,
and then \eqref{eq:set-function-identity} implies that
$q = F_v (\mathcal{U}, \cdot)$, where
$\mathcal{U} = (\mathcal{T} \cup \mathcal{R}) \setminus (\mathcal{T}
\cap \mathcal{R})$.
Uniqueness of $\mathcal{U}$ follows because there are exactly $2^m$
subsets of $\mathcal{S}$, and $2^m$ distinct maps
$r: \mathcal{S} \rightarrow \{1,-1\}$.

We now assume that $m$ is odd, and prove the two subcases. Suppose
$q: \mathcal{S} \rightarrow \{1,-1\}$ is a map such that
$f(q) \equiv f(v) \mod 2$, as stated in the lemma. If we define the
sets
\begin{equation}
\mathcal{T} = \{x \in \mathcal{S} : v(x) = 1\}, \; \mathcal{R} = \{x \in \mathcal{S} : q(x) = 1\}, \; \text{and} \; \mathcal{V} = (\mathcal{T} \cup \mathcal{R}) \setminus (\mathcal{T} \cap \mathcal{R}),
\end{equation}
then it follows that
$q = F_{F_v (\mathcal{T},\cdot)} (\mathcal{R},\cdot) = F_v
(\mathcal{V}, \cdot)$.
It is also easily checked that because $m$ is odd, we have
$F_v (\mathcal{V}, \cdot) = F_v (\mathcal{S} \setminus \mathcal{V},
\cdot)$.
The fact that $\mathcal{W} \subseteq \mathcal{S}$,
$\mathcal{W} \notin \{ \mathcal{V}, \mathcal{S} \setminus
\mathcal{V}\}$
implies $q \neq F_v (\mathcal{W}, \cdot)$, follows because there are
exactly $2^{m-1}$ distinct maps $r: \mathcal{S} \rightarrow \{1,-1\}$
with $f(r) \equiv f(v) \mod 2$. Since this exhausts all possible
subsets of $\mathcal{S}$, it also means that there does not exist any
subset $\mathcal{W} \subseteq \mathcal{S}$ with
$r = F_v (\mathcal{W}, \cdot)$, for every map
$r: \mathcal{S} \rightarrow \{1,-1\}$ such that
$f(r) \not\equiv f(v) \mod 2$. This finishes the proof.
\end{proof}

Given an anticommuting minimal generating set $\mathcal{G}$, we can
now prove the following theorem that completely characterizes the
commutativity maps on the cosets of $\langle\mathcal{G}\rangle$.

\begin{theorem}
\label{thm-anticommuting-commutativity-pattern}
Let $\mathcal{G}$ be an anticommuting minimal generating set with
$|\mathcal{G}| = m$. If $\mathcal{G} = \{I\}$, all elements of
$\mathcal{P}_n / K$ commute with $I$. Otherwise the commutativity maps
with respect to $\mathcal{G}$ of the elements in the cosets of
$\langle \mathcal{G} \rangle$, have the following structure:

\begin{enumerate}[(a)]
\setmargins
\item If $m$ is even, then in every coset of
  $\langle \mathcal{G} \rangle$, for every commutativity map
  $q : \mathcal{G} \rightarrow \{1,-1\}$ there exists exactly one
  element $P$, such that $\zet{P}{\mathcal{G}} = q$.

\item If $m$ is odd and $\mathcal{T}$ is a coset of
  $\langle \mathcal{G} \rangle$, then for all
  $Q_1,Q_2 \in \mathcal{T}$,
  $f(\zet{Q_1}{\mathcal{G}})\equiv f(\zet{Q_2}{\mathcal{G}}) \mod 2$,
  using the notation in~\eqref{eq:commuting-count-notation}. Moreover
  if $P \in \mathcal{T}$, then for every commutativity map
  $q : \mathcal{S} \rightarrow \{1,-1\}$ such that
  $f(q) \equiv f(\zet{P}{\mathcal{G}}) \mod 2$, there exist exactly
  two elements $Q_1, Q_2 \in \mathcal{T}$ with
  $Q_2 = Q_1 (\prod \mathcal{G})$, such that
  $\zet{Q_1}{\mathcal{G}} = \zet{Q_2}{\mathcal{G}} = q$; while for
  every commutativity map $q : \mathcal{G} \rightarrow \{1,-1\}$ such
  that $f(q) \not\equiv f(\zet{P}{\mathcal{G}}) \mod 2$,
  $\zet{Q}{\mathcal{G}} \neq q$ for all $Q \in \mathcal{T}$.

\item If $m$ is odd, the cosets of $\langle \mathcal{G} \rangle$ can
  be grouped into two disjoint sets $\mathcal{F}_0$ and
  $\mathcal{F}_1$, with
  $|\mathcal{F}_0| = |\mathcal{F}_1| = 2^{2n-m-1}$, such that for all
  $\mathcal{T}_0 \in \mathcal{F}_0$ and all $P_0\in\mathcal{T}_0$ it
  holds that $f(\zet{P_0}{\mathcal{G}}) \equiv 0 \mod 2$, while for
  all $\mathcal{T}_1 \in \mathcal{F}_1$ and all $P_1\in\mathcal{T}_1$
  it holds that $f(\zet{P_1}{\mathcal{G}}) \equiv 1 \mod 2$.
\end{enumerate}
\end{theorem}

\begin{proof}
The case $\mathcal{G} = \{I\}$ is obvious.

\begin{description}
\setmargins
\item{\textit{(a)}, \textit{(b)}} If $\mathcal{T}$ is a coset of
  $\langle \mathcal{G} \rangle$, then choosing any element
  $P \in \mathcal{T}$, we have
  $\mathcal{T} = P \ast \langle \mathcal{G} \rangle$. Because
  $\mathcal{G}$ is a minimal generating set, this induces a bijection
  $h : 2^{\mathcal{G}} \rightarrow \mathcal{T}$, defined by
  $h(\mathcal{U}) = P (\prod\mathcal{U})$. Given any element
  $Q \in \mathcal{T}$, we have $Q = P (\prod \mathcal{U})$ for some
  $\mathcal{U} \subseteq \mathcal{G}$. Moreover, the commutativity map
  $\zet{Q}{\mathcal{G}}$, can be expressed in terms of the
  commutativity map $\zet{P}{\mathcal{G}}$ as
\begin{equation}
\zet{Q}{\mathcal{G}}(x) = 
\begin{cases}
\zet{P}{\mathcal{G}}(x) (-1)^{|\mathcal{U}| - 1} & \;\; \text{ if } x \in \mathcal{U}, \\
\zet{P}{\mathcal{G}}(x) (-1)^{|\mathcal{U}|} & \;\; \text{ if } x \notin \mathcal{U},
\end{cases}
\end{equation}
for all $x \in \mathcal{G}$. The results then follow by applying Lemma
\ref{lem-set-function-prop}, with $v(x) =
\zet{P}{\mathcal{G}}(x)$.
For all commutativity maps $q$ that satisfy
$f(q)\equiv f(\zet{P}{\mathcal{G}}) \mod 2$, we can find the
corresponding sets using the constructions in Lemma
\ref{lem-set-function-prop}, and additionally for the odd case, by
noting that for any $\mathcal{R} \subseteq \mathcal{G}$,
$P(\prod \mathcal{R}) (\prod \mathcal{G}) = P (\prod (\mathcal{G}
\setminus \mathcal{R}))$.

\item{\textit{(c)}} Lemma \ref{lem-pattern-generators} guarantees that
  $\mathcal{P}_n / K = \mathcal{E} \sqcup \mathcal{O}$, with
  $|\mathcal{E}| = |\mathcal{O}| = 4^n / 2$, where
\[
\mathcal{E} = \{y \in \mathcal{P}_n / K : f(\zet{y}{\mathcal{G}})
\equiv 0 \mod 2 \}
\]
and
\[
\mathcal{O} = \{y \in \mathcal{P}_n / K : f(\zet{y}{\mathcal{G}}) \equiv 1 \mod 2 \}.
\]
As the cardinality of each coset is $2^m$, we get that
$|\mathcal{F}_0| = |\mathcal{F}_1| = (4^n/2) / 2^m = 2^{2n - m -
  1}$. The desired result then follows by using (b).
\end{description}\vspace*{-12pt}
\end{proof}

We now have all the necessary ingredients for an efficient randomized
algorithm to create a minimal generating set $\mathcal{G}$ that is
anticommuting and has a maximum cardinality of $2n$. The algorithm is
summarized in Algorithms~\ref{alg-egs} and~\ref{alg-aec}. We
initialize the set $\mathcal{T}$ with a given initial anticommuting
minimal generating set $\mathcal{G}$. As long as the size of
$\mathcal{T}$ is less than $2n$ we add feasible elements to it.  Once
this size has been reached, we can use
Theorem~\ref{Thm:AnticommutingBasics}(\ref{Thm:AnticommutingBasics-Q})
to get a maximally anticommuting set of size $2n+1$.

We now explain a single step of the process of adding one feasible
element to $\mathcal{T}$ --- the general case then follows by
iterating this step. We first draw an element $U$ from
$\mathcal{P}_n/K$ uniformly at random, and determine the coset
$\mathcal{S}$ of $\langle \mathcal{T} \rangle$ such that
$U \in \mathcal{S}$. We then use the results from
Theorem~\ref{thm-anticommuting-commutativity-pattern} to efficiently
find a feasible element in $\mathcal{S}$ that can be added to
$\mathcal{T}$, such that the new set is still anticommuting and a
minimal generating set.

When the size of the current set $\mathcal{T}$ is odd, such a feasible
element can be generated only if the number of elements in
$\mathcal{T}$ that commute with $U$ is even. In this case
Theorem~\ref{thm-anticommuting-commutativity-pattern} guarantees that
there are exactly two feasible elements in $\mathcal{S}$,
$\mathcal{S} \neq \mathcal{\langle \mathcal{T} \rangle}$ to choose
from, and so we choose the element that is cheaper to compute. When
the current size of $\mathcal{T}$ is even, we can always find exactly
one element in $\mathcal{S}$ that anticommutes with all elements in
$\mathcal{T}$ by
Theorem~\ref{thm-anticommuting-commutativity-pattern}. However, in
order for it to be a feasible element we must have that
$\mathcal{S}\neq \langle \mathcal{T} \rangle$. This restriction exists
to prevent the set from becoming maximal prematurely: if
$\mathcal{S} = \langle \mathcal{T} \rangle$, then the only element in
$\mathcal{S}$ that anticommutes with $\mathcal{T}$ is
$\prod\mathcal{T}$. If we succeed in finding a new element we add it
to $\mathcal{T}$, otherwise we simply redraw a new random element from
$\mathcal{P}_n/K$ and repeat the process until a feasible element is
found.
%
The complexity of determining set $\mathcal{C}$ in
Algorithm~\ref{alg-egs} is $\mathcal{O}(n\vert\mathcal{T}\vert)$, and
matches the worst-case complexity to evaluate the subsequent products
$\prod\mathcal{C}$ and $\prod(\mathcal{T}\setminus\mathcal{C})$.  The
computational complexity of Algorithm~\ref{alg-egs} is therefore
$\mathcal{O}(n\vert\mathcal{T}\vert)$.  Randomly sampling an element
from $\mathcal{P}_n/K$ takes $\mathcal{O}(n)$ time, and the element is
accepted with probability at least $1/2$. Therefore, if
$2n - \vert\mathcal{G}\vert$ is $\mathcal{O}(n)$, then the expected
runtime of Algorithm~\ref{alg-egs} is $\mathcal{O}(n^3)$.

\begin{algorithm}[t!]
\caption{Extend anticommuting minimal generating set to cardinality $2n$}\label{alg-egs}

\begin{algorithmic}[1]

    \Procedure{Extend\_Generating\_Set}{$\mathcal{G}$} \Comment{$\mathcal{G}$ is an anticommuting minimal generating set}

        \State Set $\mathcal{T} \leftarrow \mathcal{G}$, $P \leftarrow \prod \mathcal{G}$, and  $k \leftarrow |\mathcal{G}|$
        
        \While{$k < 2n$}
            
            \State $U \leftarrow$ Sample uniformly from $\mathcal{P}_n / K$ 
            \State $V \leftarrow$ {\small{ANTICOMMUTING\_ELEMENT\_COSET}}($\mathcal{T}, U$) 
            \If{(($k$ even and $V \neq P$) or ($k$ odd and $V \neq 0$))} \Comment{Acceptance criteria}
                \State $\mathcal{T} \leftarrow \mathcal{T} \cup \{V\}$, $P \leftarrow PV$
                \State $k \leftarrow k + 1$
            \EndIf
        \EndWhile
        
        \State \textbf{return} $\mathcal{T}$
    
    \EndProcedure

\end{algorithmic}
\end{algorithm}

\begin{algorithm}[t!]
\caption{Find anticommuting Pauli in coset}\label{alg-aec}

\begin{algorithmic}[1]

    \Procedure{Anticommuting\_Element\_Coset}{$\mathcal{T}, P$}
        
        \State Set $\mathcal{C} \leftarrow \{x \in \mathcal{T}: \zet{P}{\mathcal{T}}(x) = 1\}$ \Comment{Find elements in $\mathcal{T}$ that commute with $P$}
        \If{$\vert\mathcal{T}\vert$ odd and $\vert\mathcal{C}\vert$ odd}
            \State $U \leftarrow 0$ \Comment{Anticommuting element does not exist in coset}
        \ElsIf{($\vert\mathcal{T}\vert$ even and
          $\vert\mathcal{C}\vert$ even) or ($\vert\mathcal{T}\vert$
          odd and $|\mathcal{C}| \leq \lfloor |\mathcal{T}| / 2\rfloor$)}
            \State $U \leftarrow P(\prod \mathcal{C})$
        \Else
            \State $U \leftarrow P(\prod (\mathcal{T} \setminus \mathcal{C}))$
        \EndIf
        \State \textbf{return} $U$
\EndProcedure

\end{algorithmic}
\end{algorithm}

\section{Number of unique sets}\label{Sec:Counting}

In this section we consider the number of unique sets that are
maximally commuting or anticommuting, for which we derive explicit
expressions.

\subsection{Commuting sets}\label{Sec:commuting-counting}

The first lemma gives the number of ways a commuting minimal
generating set can be extended to a larger commuting minimal
generating set.

\begin{lemma}
\label{lem-counting-extensions-commuting}
Let $\mathcal{G} \subseteq \mathcal{P}_n / K$ be a commuting minimal
generating set, possibly empty, with $|\mathcal{G}| = m$. Consider the
extension to a larger commuting minimal generating set
$\mathcal{G}' \subseteq \mathcal{P}_n / K$, such that
$\mathcal{G} \subset \mathcal{G}'$, and $|\mathcal{G}'| = m'$ with
$m' \leq n$. If $\mathcal{G} = \{I\}$, no extensions are
possible. When $m' = 1$ and $m = 0$, there are, $4^n$ distinct ways to perform the
extension. Otherwise there are
$\left( \prod_{k=m}^{m'-1} (4^n / 2^k - 2^k) \right) / (m' - m)!$
distinct extensions.
\end{lemma}

\begin{proof}
If $\mathcal{G} = \{I\}$, then it cannot be extended to a minimal
generating set because of Lemma
\ref{lem-generator-abelian-pauli1}(b). For the case
$m' = 1, \; m = 0$, every singleton set is a commuting minimal
generating set, and so there are $4^n$ ways to extend $\mathcal{G}$.

For the rest of the proof we assume that $m' > 1$,
$\mathcal{G} \neq \{I\}$, and fix any arbitrary ordering of the
elements in $\mathcal{G}$. By Lemma
\ref{lem-generator-abelian-pauli1}(b), we then also have that
$I \notin \mathcal{G}$. First consider the case $m' = k + 1$, and
$m = k$. Lemma \ref{lem-pattern-generators} implies that the
cardinality of the set
$\mathcal{H} = \{P \in \mathcal{P}_n / K : \comm{P}{Q} = 1, \;
\forall \; Q \in \mathcal{G}\}$
is $4^n / 2^k$, and clearly
$\langle \mathcal{G} \rangle \subseteq \mathcal{H}$ with
$\vert\langle\mathcal{G}\rangle\vert = 2^k$. Thus the number of
distinct ways to extend $\mathcal{G}$ to $\mathcal{G'}$ is
$4^n / 2^k - 2^k$, because we can choose any element of
$\mathcal{H} \setminus \langle \mathcal{G} \rangle$.

Returning now to the case of a general $m'$, we can first order the
elements in $\mathcal{G}'$ such that the first $m$ elements are
always those of $\mathcal{G}$ in the fixed order. If we count all
the possible orderings of the remaining elements in $\mathcal{G}'$,
it follows from the previous paragraph and by noting that if
$\mathcal{G} = \emptyset$, then there are $4^n - 1$ ways to extend
$\mathcal{G}$ by one element without including $I$, that the number
of ways to extend $\mathcal{G}$ to $\mathcal{G}'$ is given by
$\prod_{k=m}^{m'-1} (4^n / 2^k - 2^k)$. Since there are exactly
$(m'-m)!$ permutations of the newly added elements, the number of
distinct extensions of $\mathcal{G}$ to $\mathcal{G}'$ is given by
$\left( \prod_{k=m}^{m'-1} (4^n / 2^k - 2^k) \right) / (m' - m)!$.
\end{proof}

The next lemma counts the number of commuting minimal generating sets
that generate the same commuting subgroup.

\begin{lemma}\label{lem-counting-commuting-generating-sets}
Let $\mathcal{S} \subseteq \mathcal{P}_n / K$, be a subgroup such
that all elements commute. By Lemma
\ref{lem-generator-abelian-pauli1}(b) and Theorem
\ref{Theorem:maximally-commuting-maximum}, $|\mathcal{S}| = 2^{m}$,
for $0 \leq m \leq n$. Then the number $N_m$ of distinct commuting
minimal generating sets $\mathcal{G}$ such that
$\langle \mathcal{G} \rangle = \mathcal{S}$ is given by
\begin{equation}
N_m = \frac{1}{m!} \prod_{k=0}^{m -1} (2^m - 2^k).
\end{equation}
\end{lemma}
\begin{proof}
The case $\mathcal{S} = \{I\}$ is obvious, so assume that
$|\mathcal{S}| \geq 2$. By Lemma
\ref{lem-generator-abelian-pauli1}(b), if $\mathcal{G}$ is a minimal
generating set of $\mathcal{S}$, then $I \notin \mathcal{G}$, and so
the first element of $\mathcal{G}$ can be chosen in $(2^m - 1)$
ways. Now suppose the first $k < m$ elements of $\mathcal{G}$ have
already been chosen. These $k$ elements form a minimal generating
set that generates a subgroup of $\mathcal{S}$ of order $2^k$. Thus
the $(k+1)$st element can be chosen in $(2^m - 2^k)$ ways. Iterating
over $0 \leq k \leq m -1$, the number of ways to form the minimal
generating set $\mathcal{G}$ using this process is
$\prod_{k=0}^{m -1} (2^m - 2^k)$. Since we do not want to count the
permutations of the elements in $\mathcal{G}$, the number of
distinct commuting minimal generating sets $\mathcal{G}$ is
$\left( \prod_{k=0}^{m -1} (2^m - 2^k) \right) / m!$.
\end{proof}

We can now easily count the number of commuting subgroups of a fixed order.

\begin{lemma}
\label{lem-counting-commuting-subgroups}
The number $N_m$ of distinct commuting subgroups of $\mathcal{P}_n / K$ of order $2^m$, for $0 \leq m \leq n$, is
\begin{equation}
N_m = \prod_{k=0}^{m-1} \frac{(4^n / 2^k - 2^k)}{(2^m - 2^k)}.
\label{eq:counting-commuting-subgroups}
\end{equation}
\end{lemma}

\begin{proof}
If $m = 0$, $\mathcal{S} = \{I\}$ is the only commuting subgroup of
order $1$, and so the statement is true. Now assume that $m > 0$,
and so by Lemma \ref{lem-generator-abelian-pauli1}(b), if
$\mathcal{G}$ is a minimal generating set of a commuting subgroup
$\mathcal{S}$ of order $2^m$, then $I \notin \mathcal{G}$. By Lemma
\ref{lem-counting-extensions-commuting}, the number of distinct ways
to form a commuting minimal generating set of cardinality $m$ is
$\left( \prod_{k=0}^{m-1} (4^n / 2^k - 2^k) \right) / m!$, where we
note that the formula is correct even when $m=1$. These generate all
possible commuting subgroups of $\mathcal{P}_n / K$ of order $2^m$,
but as each such subgroup is generated exactly by
$\left( \prod_{k=0}^{m -1} (2^m - 2^k) \right) / m!$ distinct
commuting minimal generating sets by Lemma
\ref{lem-counting-commuting-generating-sets}, we have that
\begin{equation}
N_m = \frac{\left( \prod_{k=0}^{m-1} (4^n / 2^k - 2^k) \right) / m!}{\left( \prod_{k=0}^{m -1} (2^m - 2^k) \right) / m!} = \prod_{k=0}^{m-1} \frac{(4^n / 2^k - 2^k)}{(2^m - 2^k)}.
\end{equation}
\end{proof}

We therefore have the following result (see also~\cite{SAN2007Pa} and references therein):

\begin{corollary}
\label{cor-counting-maximum-commuting-subgroups}
The number of distinct maximally commuting subgroups of
$\mathcal{P}_n / K$ is $\prod_{k=0}^{n-1} (1 + 2^{n-k})$.
\end{corollary}

\begin{proof}
The proof follows from Lemma \ref{lem-counting-commuting-subgroups}
by setting $m = n$ in \eqref{eq:counting-commuting-subgroups}.
\end{proof}

\subsection{Anticommuting sets}\label{Sec:anticommuting-counting}

We now return to the question of how many ways it is possible to
extend a minimal generating set of anticommuting abelian Paulis, which
was originally raised in
Section~\ref{Sec:AnticommutingGeneration}. The following theorem
specifies in how many ways this can be achieved so that the larger set
is still a minimal generating set.

\begin{theorem}
\label{thm-counting-extensions-anticommuting}
Let $\mathcal{G} \subseteq \mathcal{P}_n / K$ be an anticommuting
minimal generating set, possibly empty, and $|\mathcal{G}| = m$.
Consider the extension of $\mathcal{G}$ to a larger anticommuting
minimal generating set $\mathcal{G}' \subseteq \mathcal{P}_n / K$,
such that $\mathcal{G} \subset \mathcal{G}'$, with
$|\mathcal{G}'| = m'$ and $m' \leq 2n$.  If $\mathcal{G} = \{I\}$,
then it cannot be extended. When $m' = 1$ and $m = 0$, there are $4^n$
distinct ways to perform the extension. Otherwise there are
$\left( \prod_{k=m}^{m'-1} s(k) \right) / (m' - m)!$ distinct ways to
extend the set, where
\begin{equation}
s(k) = 
\begin{cases}
4^n / 2^k & \;\; \text{if } k \text{ is odd}, \\
4^n / 2^k - 1 & \;\; \text{if } k \text{ is even}.
\end{cases}
\label{eq:counting-func-extensions}
\end{equation}
\end{theorem}

\begin{proof}
If $\mathcal{G} = \{I\}$, then it is maximally anticommuting and
cannot be extended. For the case $m' = 1, \; m = 0$, every singleton
set is an anticommuting minimal generating set, and so there are
$4^n$ ways to extend $\mathcal{G}$.

For the rest of the proof we assume $m' > 1$, and
$\mathcal{G} \neq \{I\}$. Fix any arbitrary ordering of the elements
in $\mathcal{G}$. We first consider the case when $m' = k + 1$ and
$m = k$. There are exactly $4^n / 2^k$ cosets of
$\langle \mathcal{S} \rangle$. When $k$ is odd, half of the cosets
contain exactly $2$ elements that anticommute with all the elements
in $\mathcal{G}$, and the other half contain none, by
Theorem~\ref{thm-anticommuting-commutativity-pattern}(b)
and~(c). Moreover none of these elements are in
$\langle \mathcal{G}\rangle$: it follows from
equation~\eqref{eq:commuting-count-notation} that for any element
$P \in \mathcal{G} $ we have $f(\mathcal{Z}_P) \equiv 1 \mod 2$, and
so Theorem~\ref{thm-anticommuting-commutativity-pattern}(b)
applies. This gives exactly $4^n / 2^k$ distinct ways to extend
$\mathcal{G}$. If $k$ is even, then by Theorem
\ref{thm-anticommuting-commutativity-pattern}(a) each coset contains
exactly one element that anticommutes with all the elements in
$\mathcal{G}$, and so there are exactly $4^n / 2^k - 1$ distinct
ways to extend $\mathcal{G}$ (counting every element in each coset,
except $\prod \mathcal{G}$), because if we include
$\prod \mathcal{G} \in \langle \mathcal{G} \rangle$, then
$\mathcal{G} \cup \{\prod \mathcal{G}\}$ is not a minimal generating
set.

We now return to the case of a general $m'$. If
$\mathcal{G} = \emptyset$, then there are $4^n - 1$ ways to
initialize $\mathcal{G}'$, since we have to exclude $I$. It follows
from the previous paragraph that adding the $(k+1)$-th element for
$k > 0$, can be done in $s(k)$ distinct ways, with $s(k)$ defined as
in~\eqref{eq:counting-func-extensions}.  Since there are exactly
$(m'-m)!$ permutations of the new elements, the number of distinct
extensions of $\mathcal{G}$ to $\mathcal{G}'$ is given by
$\left( \prod_{k=m}^{m'-1} s(k) \right) / (m' - m)!$.
\end{proof}

As a result of the previous theorem, we can now immediately count the
number of maximally anticommuting sets. This is carried out in the
next corollary.
\begin{corollary}
If $N_m$ is the number of maximally anticommuting subsets of
$\mathcal{P}_n / K$ of cardinality $m$, then using $s(k)$ as defined
in \eqref{eq:counting-func-extensions}
\begin{equation}
\label{eq:counting-func-maximally-anticommuting}
N_m = 
\begin{cases}
\frac{1}{m!} \prod\limits_{k=0}^{m-2} s(k) & \;\; \text{if } m \text{ odd, and } m \leq 2n+1, \\
0 & \;\; \text{otherwise}.
\end{cases}
\end{equation}
\end{corollary}

\begin{proof}
By Theorem
\ref{Thm:AnticommutingBasics}(\ref{Thm:AnticommutingBasics-n-odd})
and Lemma \ref{lem-max-anticommuting-paulis}, the assertion is
clearly true when $m$ is even, or when $m > 2n+1$. The statement is
also true for $m = 1$, as $\{I\}$ is the only such set. Now let
$\mathcal{G} \subseteq \mathcal{P}_n / K$ be a maximally
anticommuting set with $|\mathcal{G}| = m$, $m$ odd, and
$3 \leq m \leq 2n+1$. Then there are $m$ distinct ways to remove an
element from $\mathcal{G}$ to obtain a minimal generating
set. Noting that $I \notin \mathcal{G}$ when $m \geq 3$, the
statement then follows as there are
$\left( \prod_{k=0}^{m-2} s(k) \right)/ (m-1)!$ distinct
anticommuting minimal generating sets set of cardinality $m - 1$, by
Theorem \ref{thm-counting-extensions-anticommuting}.
\end{proof}

\section{Discussion}\label{Sec:Conclusions}

An interesting class of commuting and anticommuting Paulis is one in
which all Paulis are formed exclusively by Kronecker products of the
three Pauli matrices $\sigma_x$, $\sigma_y$, and
$\sigma_y$. Obviously, this restriction limits the maximum possible
size of commuting sets; at the very least, the $I$ element is no
longer present. The size of anticommuting sets is also affected, and
whereas the maximum size of $2n+1$ can be attained for $n$-Paulis for
$n=1$ and $n=4$:
\[
\mathcal{M}_1 = \left[\begin{array}{ccc}
x&y&z\end{array}\right],\qquad
\mathcal{M}_4 = \left[\begin{array}{ccccccccc}
x&x&x&y&y&y&z&z&z\\
x&y&z&x&y&z&x&y&z\\
x&y&z&y&z&x&z&x&y\\
x&y&z&z&x&y&y&z&x\\
\end{array}\right],
\]
it is easy to show that this is not the case for $n=2$ and
$n=3$. Indeed, it follows from Theorem~\ref{thm-maximum-anticommuting}
that any anticommuting set $\mathcal{G}$ of maximum size consisting of
only $\sigma_x$, $\sigma_y$, and $\sigma_z$ terms, must satisfy
$\vert\mathcal{C}_{\ell}\vert \geq 3$ for all $\ell\in\{x,y,z\}$,
since otherwise these sets could not multiply to $I$. For the set to
be maximal, we therefore require $k \geq 9$, but this exceeds the
maximum possible size of $2n+1$ for these values of $n$.

An alternative formulation of commutativity and anticommutativity is
to ask for sets of vectors in $\mathrm{GF}(3)^n$ such that the Hamming
distance between every pair of vectors is even or odd,
respectively. The question of how large such sets can be is studied
in~\cite{ALO2007La,ALO2007Lb}. They show that the asymptotic size is
$\Theta(2^n)$ for commuting sets, and $\Theta(n)$ for anticommuting
sets, but also provide more specific bounds. For instance,
\cite[Corollary 2.10]{ALO2007La} shows that for $n\geq 1$, any
commuting subset $\mathcal{G}_{\mathrm{xyz}} \subset \mathcal{P}_n/K$
satisfies
\[
\vert\mathcal{G}_{\mathrm{xyz}}\vert \leq \frac{2^n}{1 - \left(-\frac{1}{3}\right)^n + \left(\frac{2}{3}\right)^n}.
\]
Numerical values for the maximum possible sizes of restricted
anticommuting sets with $n\in[8]$ are given in~\cite{A128036}. For
$n = 8$ this shows that it is again possible to attain the maximum
size of $2n+1$, and indeed we have
\[
\mathcal{M}_8 = \left[\begin{array}{ccccccccccccccccc}
x&x&x&x&x&x&x&y&y&y&z&z&z&z&z&z&z\\
x&x&x&x&x&y&z&x&y&z&x&y&z&z&z&z&z\\
x&x&x&x&x&y&z&y&z&x&z&x&y&y&y&y&y\\
x&x&x&x&y&x&z&z&y&x&x&z&y&z&z&z&z\\
x&x&x&y&x&z&x&x&y&z&z&x&z&y&z&z&z\\
x&y&z&x&y&x&y&y&z&x&x&y&x&y&x&y&z\\
x&y&z&y&z&y&z&z&x&y&y&z&y&z&z&x&y\\
x&y&z&z&x&z&x&x&y&z&z&x&z&x&y&z&x\\
\end{array}\right].
\]
The values of $n$ for which these restricted anticommuting sets can
attain the maximum size remains an open question.

\section*{Acknowledgments}

The authors would like to thank Sergey Bravyi, Kristan P.~Temme, and
Theodore J.~Yoder for useful discussions. R.S.~would like to thank IBM
T.J.~Watson Research Center for facilitating the research.

\bibliographystyle{unsrt}
\bibliography{bibliography}

\begin{thebibliography}{10}

\bibitem{gottesman1997phd}
Daniel~E. Gottesman.
\newblock {\em Stabilizer Codes and Quantum Error Correction. Dissertation}.
\newblock PhD thesis, California Institute of Technology, 1997.
\newblock \url{http://resolver.caltech.edu/CaltechETD:etd-07162004-113028}.

\bibitem{bravyi2002fermionic}
Sergey~B. Bravyi and Alexei~{Yu}. Kitaev.
\newblock Fermionic quantum computation.
\newblock {\em Annals of Physics}, 298(1):210--226, 2002.

\bibitem{calderbank2001orthogonal}
AR~Calderbank and AF~Naguib.
\newblock Orthogonal designs and third generation wireless communication.
\newblock {\em London Mathematical Society Lecture Note Series}, pages 75--107,
  2001.

\bibitem{ROT1999a}
Joseph~J. Rotman.
\newblock {\em An introduction to the theory of groups}.
\newblock Graduate texts in mathematics. Springer, 4th edition, 1999.

\bibitem{hrubes2016families}
Pavel Hrube\v{s}.
\newblock On families of anticommuting matrices.
\newblock {\em Linear Algebra and its Applications}, 493:494--507, 2016.

\bibitem{bonet2019nearly}
Xavier Bonet-Monroig, Ryan Babbush, and Thomas~E O'Brien.
\newblock Nearly optimal measurement scheduling for partial tomography of
  quantum states.
\newblock {\em arXiv preprint arXiv:1908.05628}, 2019.

\bibitem{SAN2007Pa}
Metod Saniga and Michel Planat.
\newblock Multiple qubits as symplectic polar spaces of order two.
\newblock {\em Advanced Studies in Theoretical Physics}, 1(1):1--4, 2007.

\bibitem{ALO2007La}
Noga Alon and Eyal Lubetzky.
\newblock Codes and {X}or graph products.
\newblock {\em Combinatorica}, 27(1):13--33, 2007.

\bibitem{ALO2007Lb}
Noga Alon and Eyal Lubetzky.
\newblock Graph powers, {D}elsarte, {H}offman, {R}amsey, and {S}hannon.
\newblock {\em SIAM Journal on Discrete Mathematics}, 21(2):329--348, 2007.

\bibitem{A128036}
Neil~J. Sloane.
\newblock The on-line encyclopedia of integer sequences.
\newblock Sequence A128036, \url{https://oeis.org/A128036}.

\end{thebibliography}

\end{document}